\documentclass[journal,onecolumn]{IEEEtran}

\usepackage{amsmath}
\usepackage{amssymb}
\usepackage{amsthm}
\usepackage{mathrsfs}
\usepackage{float}
\usepackage{graphicx}
\usepackage{xcolor}
\usepackage{cite}
\usepackage[colorlinks=true,citecolor=blue,urlcolor=blue,linkcolor=blue]{hyperref}
\usepackage[nameinlink,capitalize]{cleveref}
\crefname{enumi}{}{}
\usepackage[utf8]{inputenc}
\DeclareUnicodeCharacter{200B}{}
\DeclareUnicodeCharacter{FB01}{fi}

\theoremstyle{plain}
\newtheorem{theorem}{Theorem}[section]
\newtheorem{lemma}[theorem]{Lemma}
\newtheorem{proposition}[theorem]{Proposition}
\newtheorem{corollary}[theorem]{Corollary}
\theoremstyle{definition}
\newtheorem{definition}[theorem]{Definition}

\theoremstyle{remark}
\newtheorem{remark}[theorem]{Remark}
\usepackage{comment}

\DeclareMathOperator{\rank}{rank}

\newcommand{\hH}{\mathcal{H}}
\newcommand{\PG}{\operatorname{PG}}
\numberwithin{equation}{section}
\newcommand{\bc}{\mathbf{c}}

\newcommand{\F}{\mathbb{F}}

\newcommand{\C}{\mathcal{C}}
\newcommand{\one}{\mathbf{1}}

\newcommand{\yrev}[1]{\mbox{#1}}

\usepackage[shortlabels, inline]{enumitem}
\setlist[enumerate,1]{label=\textup{(\arabic*)}}
\SetEnumerateShortLabel{a}{\textup{(\alph*)}}
\SetEnumerateShortLabel{A}{\textup{(\Alph*)}}
\SetEnumerateShortLabel{1}{\textup{(\arabic*)}}
\SetEnumerateShortLabel{i}{\textup{(\roman*)}}
\SetEnumerateShortLabel{I}{\textup{(\Roman*)}}

\begin{document}

\title{Minimality of Binary Codes From Weight Layers and Graph Cut Spaces}

\author{Xianhong~Xie, Mengni~Wang and~Tingting~Wu%
\thanks{X. Xie, M. Wang, and T. Wu are with the School of Artificial Intelligence, Anhui Agricultural University, Hefei 230036, China (e-mail:
xianhxie@ahau.edu.cn; wangmn@stu.ahau.edu.cn; wutingting58@163.com).}%
\thanks{X. Xie is also with the State Key Laboratory of Cyberspace Security
Defense, Institute of Information Engineering, Chinese Academy of Sciences,
Beijing 100085, China.}%
\thanks{Corresponding author: Tingting Wu.}}

\markboth{IEEE Journal Article}%
{Xie \MakeLowercase{\textit{et al.}}: Minimality of Binary Codes From Weight Layers and Graph Cut Spaces}

\maketitle

\begin{abstract}
	Minimal binary linear codes play an important role in secret-sharing schemes and in the structural study of linear codes. We study the minimality of three classes of known constructions: codes from a single complete Hamming-weight layer, codes from arbitrary unions of such layers and graph cut-space codes. For the previously studied single-layer codes \(\mathcal{C}_{k,t}\), we give a complete minimality classification: \(\mathcal{C}_{k,t}\) is minimal if and only if \(t\) is even and \(2\le t\le k-2\) or \((k,t)=(3,2)\). For \(t=4\), we determine the minimum and maximum weights, and show that the family strictly violates the binary Ashikhmin--Barg sufficient condition for every \(k\ge16\). We also calculate the minimum distance of the dual code. For arbitrary unions of complete Hamming-weight layers, we obtain an if-and-only-if classification of minimality and uniform criteria for evenness, self-orthogonality, the doubly-even property and hull dimension. We further analyze the family defined by \(D_3\cup D_{k-1}\) for odd \(k\). Finally, using the classical bond characterization of cuts, we translate minimality of graph cut-space codes into a graph-theoretic condition: a nonempty simple graph has a minimal cut code precisely when it has one edge-containing component and that component is complete. 
We also give a graph-function realization of the known binary tetrahedron code and explicitly identify the resulting code equivalence.
\end{abstract}


\IEEEpeerreviewmaketitle

\section{Introduction}
Let \(\mathbb{F}_q\) be the finite field of order \(q\), and let
\(\mathcal{C}\subseteq\mathbb{F}_q^n\) be a nonzero linear code. For
\(\mathbf{c}=(c_1,\ldots,c_n)\in\mathbb{F}_q^n\), its support and
Hamming weight are defined by
\[
\operatorname{supp}(\mathbf{c})
=\{i\in\{1,\ldots,n\}:c_i\ne0\},
\qquad
\operatorname{wt}(\mathbf{c})
=\#\operatorname{supp}(\mathbf{c}).
\]
\begin{definition}
  A nonzero codeword \(\mathbf{c}\in\mathcal{C}\) is called
	\emph{minimal} if, for every nonzero codeword
	\(\mathbf{c}'\in\mathcal{C}\),
	\[
	\operatorname{supp}(\mathbf{c}')
	\subseteq\operatorname{supp}(\mathbf{c})
	\quad\Longrightarrow\quad
	\mathbf{c}'=\lambda\mathbf{c}
	\]
	for some \(\lambda\in\mathbb{F}_q^{*}\).
	The code \(\mathcal{C}\) is called minimal if every nonzero codeword of
	\(\mathcal{C}\) is minimal.
\end{definition}
Minimal linear codes play an important role in coding theory and
cryptography. Massey \cite{massey1993} showed that the minimal codewords of a linear
code determine the access structure of the corresponding linear
secret-sharing scheme. This connection was subsequently developed for
several classes of linear codes and associated secret-sharing
schemes \cite{Car-din-yu,Yuan-ding}. Consequently, the construction and classification of
minimal codes, particularly those with explicitly determined
parameters and weight distributions, remain important problems.

A widely used sufficient condition for minimality is due to Ashikhmin
and Barg \cite{Ash-abrg}.
\begin{theorem}[Ashikhmin--Barg]\label{thm:ABbound}
Let $\C$ be a $q$-ary linear code, and put
\[
w_{\min}=\min_{\bc\in\C\setminus\{0\}}\operatorname{wt}(\bc),
\qquad
w_{\max}=\max_{\bc\in\C}\operatorname{wt}(\bc).
\]
If $w_{\min}/w_{\max}>(q-1)/q$, then $\mathcal C$ is minimal.
\end{theorem}
Except when general results are recalled, all codes considered in this
paper are binary. In this case, the Ashikhmin--Barg sufficient condition
becomes
\[
\frac{w_{\min}}{w_{\max}}>\frac12.
\]
This condition is sufficient but not necessary. Accordingly, the
construction of binary minimal codes satisfying
\[
\frac{w_{\min}}{w_{\max}}\le\frac12
\]
has become an active research direction. Throughout this paper, a
nonzero binary code is said to violate the Ashikhmin--Barg sufficient
condition when \(w_{\min}/w_{\max}\le1/2\). When the ratio is strictly
less than \(1/2\), we call the violation strict. The
literature most closely related to the present work can be organized
into three broad directions.

\emph{Boolean- and vectorial-function constructions.}
	Ding et al.\ \cite{Ding-heng-zhou} established a support-inclusion criterion for binary
	minimal codes, and Heng et al.\ \cite{heng-ding-zhou} extended the corresponding approach
	to codes over general finite fields. Subsequent constructions used
	characteristic and root Boolean functions \cite{Pasalic-Rodrguez,Mesnager-Qi,Li-yue,Xu-Lu-Cao,Tang-Carlet,Xie-Ouyang}, bent concatenations
	and derivatives of Maiorana--McFarland functions \cite{zhang-pasalic}, weakly regular
	plateaued functions \cite{Mesnager-Sinak}, vectorial Boolean functions \cite{Li-peng},
	two-to-one functions \cite{Mesanager-Qian}, and direct sums of functions together with
	subspaces of derivatives \cite{Rodriguez-Aldama}. In particular, some recent
	function-based constructions yield few-weight binary codes that are
	simultaneously minimal, self-orthogonal, and outside the
	Ashikhmin--Barg sufficient condition \cite{Li-chen-jin,Shi,jin-li-Qu}. These results demonstrate
	the flexibility of cryptographic functions in minimal-code
	construction, although their proofs generally depend on the Walsh
	spectra or other algebraic properties of the functions involved.

\emph{Geometric and cutting-blocking-set constructions.}
Bonini and Borello introduced blocking-set methods into the study of
minimal linear codes~\cite{Bartoli-Bonini,Bonini-Borello}. Alfarano et al.\ \cite{Alfarano-Borello} subsequently established a geometric characterization of minimal
codes and investigated their asymptotic performance. Related combinatorial interpretations were
developed by Alfarano et al.\ \cite{Alfarano-Borello-eri}. Tang et al.\ \cite{Tang-Qiu} proved the
full characterization that a projective linear code is minimal if
and only if its projective system is a cutting blocking set.

Further constructions and developments include strong blocking sets
~\cite{Heger-Nagy}, partial difference sets~\cite{ran-feng}, explicit
constructions of optimal blocking sets~\cite{Bishnoi-Tomon}, and
minimal codes obtained from arbitrary projective linear codes
~\cite{Chen-Chen}. These results provide a broad geometric and
combinatorial framework for minimal codes, while explicit code
parameters depend on the structure of the underlying projective
system.

\emph{Simplicial-complex and defining-set constructions.}
Several related constructions must be distinguished. Chang and
Hyun~\cite{Chang-Hyun} associated to the characteristic function
of the nonzero part of a simplicial complex a Boolean-function code
whose coordinates run through all nonzero vectors, the additional
Boolean-function coefficient makes that construction different from
the defining-set code \(\mathcal C_D\) used here. Hyun et al.\ \cite{Hyun-Lee-lee} used the same generic defining-set construction
as ours, but took the complement \(\Delta^c\) of a simplicial complex
as the defining set. Wu and Lee~\cite{Wu-lee} considered differences
\(\Delta_1\setminus\Delta_2\) of simplicial complexes and explicitly
observed that complete Hamming-weight layers are contained in this
framework.

More precisely, put
\[
\Delta_{\le r}=\{x\in\mathbb F_2^k:\operatorname{wt}(x)\le r\}.
\]
Then
\[
D_t=\Delta_{\le t}\setminus\Delta_{\le t-1},
\qquad
\bigcup_{j=1}^{r}D_j=\Delta_{\le r}\setminus\{0\}.
\]
Thus neither the generic evaluation mechanism nor the single-layer and
initial-layer defining sets are new.  The distinction relevant to this
paper is that we treat every permutation-invariant rank-selected set
\(D_S=\bigcup_{t\in S}D_t\), including nonconsecutive \(S\), and seek
uniform if-and-only-if minimality results rather than a new
simplicial-complex construction.

The complete single-layer construction itself is not new. Zhou et al.\ \cite{ZhouLiTangDing2019} determined its dimension,
Gram matrix, hull behavior, LCD and self-orthogonality properties, and
gave Krawtchouk-polynomial weight formulas. Zhang at al.\ \cite{ZhangYanWei2019} treated the weight-two member as a minimal
code. More generally, incidence codes of graphs and hypergraphs have
been studied in~\cite{Maji-Mesnager,DankelmannKeyRodrigues2013,MallikYildiz2024}, and
binary graph incidence codes are the classical cut spaces whose duals
are cycle spaces. Accordingly, the contribution below is formulated
in terms of minimality classifications, exact thresholds, and a
refinement of a previously stated dual-distance condition, rather than
as the introduction of new defining-set or cut-space constructions.

The three directions above provide complementary approaches to the
study of minimal codes. Function-based methods typically derive
minimality from spectral or algebraic properties, geometric methods
describe minimality through projective point sets, and
simplicial-complex constructions exploit the combinatorial structure
of defining sets. Some recent function-based families simultaneously
possess minimality, self-orthogonality, few-weight properties, and
violation of the Ashikhmin--Barg sufficient condition
~\cite{Li-chen-jin,jin-li-Qu}. This suggests studying whether several
code invariants can be determined uniformly within a defining-set
family having a particularly simple combinatorial structure.
Accordingly, we address the following question: can one construct and
classify minimal codes while simultaneously determining their weight
distributions, dual distances, evenness, self-orthogonality, double
evenness, and hull dimensions by means of elementary combinatorial
arguments?

To address this question, we work within a common binary
evaluation-code framework. For a nonempty defining set
\(D\subseteq\mathbb{F}_2^k\setminus\{\mathbf{0}\}\), let
\[
\mathcal{C}_D
=
\left\{
\mathbf{c}_{a}
=
\langle a,x\rangle_{x\in D}:
a\in\mathbb{F}_2^k
\right\}.
\]
We investigate the following three classes of structured defining
sets.

\begin{itemize}
	\item The single Hamming-weight layer
	\[
	D_t
	=
	\{x\in\mathbb{F}_2^k:
	\operatorname{wt}(x)=t\}.
	\]
	
	\item The union of complete Hamming-weight layers
	\[
	D_S
	=
	\bigcup_{t\in S}D_t,
	\qquad
	\varnothing\ne S\subseteq\{1,2,\ldots,k-1\}.
	\]
	
	\item For a simple graph $G=(V(G),E(G))$, the associated weight-two defining set is
	\[
	D_G
	=
	\{\mathbf{e}_u+\mathbf{e}_v:uv\in E(G)\}.
	\]
\end{itemize}

These are three instances of the same evaluation construction rather
than three unrelated coding methods. Their combinatorial symmetry
reduces the relevant calculations to Krawtchouk-type sums, parity of
binomial coefficients, symmetric differences of subsets, binary
incidence matrices, and hyperplane-spanning conditions. The main
contributions are as follows.
\begin{enumerate}[1]
  \item For the previously studied single-layer codes
\(\mathcal{C}_{k,t}=\mathcal{C}_{D_t}\), we recall the dimension and
weight enumerator needed below and prove the complete minimality classification
\[
\mathcal{C}_{k,t}\text{ is minimal}
\Longleftrightarrow
t\text{ is even and }2\le t\le k-2,\
	\text{or }(k,t)=(3,2).
\]
For \(t=4\), we determine the minimum and maximum nonzero weights and prove that
\(\mathcal{C}_{k,4}\) is minimal and violates the Ashikhmin--Barg sufficient condition
for every $k\ge16$.
We also give the precise dual-distance condition, including the
feasibility requirement \(k\ge3t/2\) omitted from the corresponding
even-\(t\) assertion in~\cite{ZhouLiTangDing2019},
and derive explicit binomial-parity criteria for evenness,
self-orthogonality, and the doubly-even property.

\item For the multi-layer codes
\(\mathcal{C}_{k,S}=\mathcal{C}_{D_S}\),
we determine the dimension and weight enumerator for an arbitrary
nonempty layer-index set \(S\), and then classify minimality for all
such \(S\).
For \(k\ge4\), we prove that \(\mathcal{C}_{k,S}\) is minimal if and
only if either \(S\) contains an even integer \(t\) with
\(2\le t\le k-2\), or \(k\) is odd,
\[
k-1\in S
\qquad\text{and}\qquad
S\cap\{3,5,\ldots,k-2\}\ne\varnothing.
\]
We also obtain uniform criteria for evenness, self-orthogonality, and
the doubly-even property, together with an explicit formula for the hull
dimension. In particular, for
\[
D_S=D_3\cup D_{k-1},
\qquad k\ge5\text{ odd},
\]
the resulting code is minimal, and its minimum distance, weight
enumerator, dual distance, number of weight-three dual codewords, and
hull dimension are determined explicitly. This family strictly
violates the binary Ashikhmin--Barg sufficient condition for every odd
\(k\ge7\).

\item For the classical cut code \(\mathcal{C}_G\) of a simple graph
\(G\) with at least one edge, the bond characterization of minimal
nonempty cuts yields the following coding-theoretic classification:
\[
\mathcal{C}_G\text{ is minimal}
\quad\Longleftrightarrow\quad
G\cong K_r\sqcup(v_G-r)K_1
\quad\text{for some }r\ge2,
\]
where \(v_G=\#V(G)\). Thus \(G\) has exactly one edge-containing
connected component, and that component is complete. For this
surviving class, we record the dimension, exact weight enumerator,
minimum and maximum nonzero weights, dual distance, evenness, and hull
dimension, while identifying which of these facts are classical. The binary Ashikhmin--Barg sufficient condition is violated
exactly when \(r\ge7\): equality
\(w_{\min}/w_{\max}=1/2\) occurs for \(r=7\), while the violation is
strict for \(r\ge8\). 
Finally, we identify the complete-graph and weight-layer descriptions
of the binary tetrahedron code and give an explicit graph-function
realization obtained from a linear section. This last result is a new
representation of a known code, and no new code-equivalence class is
claimed.
\end{enumerate}

The remainder of the paper is organized as follows. Section~\ref{sec:preliminaries}
introduces the evaluation-code framework and recalls the tools used for
support containment, cutting blocking sets, binomial parity, dual
distance, and hull dimension. Section~\ref{sec:a-sin-weight} studies codes defined by a
single Hamming-weight layer. Section~\ref{sec:union-weight-layer} treats unions of complete
Hamming-weight layers. Section~\ref{sec:graph-cut-space} first
classifies minimal graph cut codes and determines their principal code
parameters, and then relates the complete-graph construction to known
defining-set codes and gives a graph-function realization.
Section~\ref{sec:conclusion}
summarizes the results and presents directions for further study.

\section{Preliminaries}
\label{sec:preliminaries}
Throughout this section, let $V$ be a finite-dimensional $\F_2$-vector space equipped with a nondegenerate bilinear form $\langle\cdot,\cdot\rangle$. Whenever a coordinate realization $V=\F_2^m$ is fixed, $\operatorname{wt}(v)$ denotes the Hamming weight of $v$.

For a finite set $X$, $\#X$ denotes its cardinality, and $\langle X\rangle$ denotes its linear span over $\F_2$.

For $S\subseteq V$, the characteristic function $1_S\in\F_2^{\#V}$ is defined by
\[1_S(a)=\begin{cases}
           1, & \mbox{if } a\in S \\
           0, & \mbox{otherwise}.
         \end{cases}\]

For any positive integer $k$, $[k]:=\{1,2,\ldots,k\}$.

For a matrix $A$, $A^T$ denotes its transpose. In particular, the transpose of a row vector is a
column vector.
\subsection{Defining-set codes}

Let
$D\subseteq V\setminus\{0\}$
be a nonempty defining set. The linear code associated with $D$ is
\[\C_D=\bigl\{\bc_a=(\langle a,x\rangle)_{x\in D}:a\in V\bigr\}\subseteq \mathbb F_2^{\#D}.\]
The evaluation map
\[\phi_D:V\longrightarrow \mathbb F_2^{\#D},\qquad a\mapsto (\langle a,x\rangle)_{x\in D}\]
has kernel
\[\ker \phi_D= \langle D\rangle^\perp=\{a\in V:\langle a,x\rangle=0,\ \forall x\in D\},\]
so that
\[\dim \mathcal C_D=\dim V-\dim(\langle D\rangle^\perp).\]
This framework contains the three constructions studied in this paper as special cases.

\begin{enumerate}[1]
	\item {\bf Single Hamming--weight layer.}
	Take $V=\F_2^k$ with the standard inner product and
	\[D=D_t=\{x\in\mathbb F_2^k:\operatorname{wt}(x)=t\}.\]
	Then
	$\C_{k,t}:=\C_{D_t}$ is the weight-layer code of Section~\ref{sec:a-sin-weight}.
	\item {\bf Multi--layer unions.}
	For a nonempty set $S\subseteq\{1,\dots,k-1\}$, let
	\[D_S=\bigcup_{t\in S}D_t.\] Then $\C_{k,S}:=\C_{D_S}$ is the multi-layer code of Section~\ref{sec:union-weight-layer}.
	\item {\bf Graph cut space.}
	Let $G=(V(G),E(G))$ be a simple graph and take
	\[V=\mathbb F_2^{\#V(G)},\qquad D=\{e_u+e_v:uv\in E(G)\}.\]
	where $\{e_i\}_{i=1}^{\#V(G)}$ is standard basis for $\F_2^{\#V(G)}$, i.e.,
	$\langle a,e_u+e_v\rangle=a_u+a_v$ for $a\in\mathbb F_2^{\#V(G)}$.
	Hence $\mathcal C_G:=\mathcal C_D$ is exactly the cut space of $G$, see Section~\ref{sec:graph-cut-space}.
\end{enumerate}
Thus the paper is not about three unrelated methods, but about
	three instances of the same evaluation construction:
minimal codes from structured defining sets $D\subseteq V$.

\subsection{Self-orthogonality and hull dimension}
Let $\mathcal{C} \subseteq \mathbb{F}_2^n$ be a binary linear code with an $m\times n$ generator matrix $G$.
Recall that $\C$ is \emph{self-orthogonal} if $\C\subseteq \C^{\perp}$, equivalently, if $GG^T=0$ over $\F_2$. In particular, every codeword of a self-orthogonal binary code has even weight.

Define the linear map
\[
\phi:\F_2^m\longrightarrow\C,\qquad a\longmapsto aG.
\]
Its kernel is $\ker\phi=\{a \in \mathbb{F}_2^m : aG = 0\}$, and $\dim(\ker\phi) = m - \dim\mathcal{C}$. The \emph{hull} of $\mathcal{C}$ is defined as
\[\hH(\mathcal{C}) = \mathcal{C} \cap \mathcal{C}^\perp.\]
A standard dimension formula relates the hull to the row Gram matrix $GG^{T} \in \mathbb{F}_2^{m \times m}$
\[
\dim\hH(\C)=\dim\C-\rank(GG^T)
=m-\dim(\ker\phi)-\rank(GG^T)
.
\]
In particular, when $G$ has full row rank ($m = \dim\mathcal{C}$), $\ker\phi = \{0\}$ and thus
\begin{equation}\label{eq:hull-fullrank}
\dim \hH(\mathcal{C}) = m - \rank(GG^{T}) = \dim\bigl(\ker(GG^{T})\bigr).
\end{equation}
Throughout this paper, the Gram matrix is the row Gram matrix $GG^T$,
rather than the matrix of pairwise inner products of the defining columns.
\smallskip

For the defining-set code $\C_D$, its self-orthogonality is as follows.
\begin{lemma}\label{lem:even-selforth-nat} Let \(D\subseteq\F_2^k\setminus\{0\}\) be nonempty, and use
the standard inner product in the definition of $\C_D$.
			Then
			\[
			\C_D\text{ is even}
			\iff\sum_{x\in D}x=0,
			\qquad
			\C_D\text{ is self-orthogonal}
			\iff\sum_{x\in D}xx^T=0.
			\]
\end{lemma}

\begin{proof}
		For \(a\in\mathbb F_2^k\), write
		\[
		\bc_a=(\langle a, x\rangle)_{x\in D}\in\mathcal C_D.
		\]
		Its Hamming weight modulo 2 satisfies
		\[
		\operatorname{wt}(\bc_a)
		\equiv
		\sum_{x\in D}\langle a,x\rangle
		=
		\left\langle a,\sum_{x\in D}x\right\rangle
		\pmod2.
		\]
		Thus every codeword of \(\mathcal C_D\) has even weight if and only if
		\(\sum_{x\in D}x=0\).
	
		Considering the natural generator matrix \(G_{\mathrm{nat}}=(x)_{x\in D}\), whose columns
		are indexed by the elements of \(D\), we have
		\[
		G_{\mathrm{nat}}G_{\mathrm{nat}}^T=\sum_{x\in D}xx^T.
		\]
		Hence \(\mathcal C_D\) is self-orthogonal if and only if
		\(G_{\mathrm{nat}}G_{\mathrm{nat}}^T=0\), equivalently,
		\[
		\sum_{x\in D}xx^T=0.\qedhere
		\]
	\end{proof}
\begin{remark}
 Suppose $\C_D$ has length $\#D$ and dimension $m$. Two distinct notions of generator matrices arise in this context:

\begin{enumerate}
    \item \textbf{Natural generator matrix.} ${G_{\mathrm{nat}}} = (x)_{x \in D}$: a $k \times \#D$ matrix whose columns are indexed by the elements of $D$. Its row space is exactly $\mathcal{C}_D$.

    \item \textbf{General generator matrix.} $G$: an $m \times \#D$ matrix whose rows form a basis of $\mathcal{C}_D$. Such a matrix need not reflect the set $D$ in its column indexing.
\end{enumerate}

Although both matrices generate the same code, they differ in dimension and geometric interpretation:
\begin{itemize}
    \item ${G_{\mathrm{nat}}}$ may not have full row rank, whereas the rows of $G$ form a basis and hence are linearly independent.
    \item The row Gram matrix of ${G_{\mathrm{nat}}}$ has size $k \times k$ and captures the pairwise outer products of elements in $D$, the row Gram matrix of $G$ has size $m \times m$ and is used for computing $\dim \hH(\mathcal{C}_D)$.
    \item The hull dimension formula applies consistently to both:
    \[
    \begin{aligned}
    \dim \hH(\mathcal{C}_D)
    =\dim \mathcal{C}_D-\rank(GG^{T})=\dim \mathcal{C}_D-\rank(G_{\mathrm{nat}}G_{\mathrm{nat}}^{T}).
    \end{aligned}
    \]
    Indeed, if $G_{\mathrm{nat}}=AG$ with $G$ a full-row-rank
    generator matrix, then $A$ has full column rank and
    $G_{\mathrm{nat}}G_{\mathrm{nat}}^T=AGG^TA^T$, hence the two Gram
    matrices have the same rank.
\end{itemize}
This distinction is crucial when transitioning from the abstract hull theory to the concrete analysis of cut codes, where $D$ corresponds to the set of edge indicator vectors.
\end{remark}

\subsection{Projective systems and cutting blocking sets}\label{subsec:PSCBS}
Let $q$ be a prime power.
Choose a full-row-rank generator matrix $G\in\F_q^{k\times n}$ for an $[n,k]$ linear code $\C$, assume first that it has no zero columns, and denote its columns by $p_1,\ldots,p_n$. The code is then
\[\mathcal C=\{(\langle x,p_1\rangle,\dots,\langle x,p_n\rangle):x\in\mathbb F_q^k\}.\]
	The corresponding projective points are recorded with multiplicities when proportional columns occur. After proportional columns are identified, the distinct points form a set $P\subseteq\PG(k-1,\F_q)$, and the cutting property depends only on $P$. If the columns are pairwise nonproportional, $P$ is the projective system of the projective code $\C$.

\begin{definition} A set $X\subseteq \PG(k-1,\F_q)$ is a cutting blocking set if for every hyperplane $H\subseteq \PG(k-1,\F_q)$,
\[\langle X\cap H\rangle_{\F_q}=H,\]
i.e., the points of $X$ lying in $H$ span $H$.\end{definition}

The following fundamental equivalence connects coding theory and finite geometry.

\begin{theorem}[Tang--Qiu--Liao--Zhou~\cite{Tang-Qiu}]\label{thm:tang-qiu-zhou} A projective linear code $\mathcal C$ is minimal if and only if its projective system is a cutting blocking set.\end{theorem}

\subsection{Lucas's and Kummer's theorems}

Many minimality and self-orthogonality results in this paper reduce to the parity or the 2-adic valuation of binomial coefficients. The principal tools are Lucas's \cite{lucas} and Kummer's theorems.
\begin{theorem}[Lucas] Let $p$ be a prime and $m,n\ge 0$ be integers with base-$p$ expansions $m=\sum m_ip^i$ and $n=\sum n_ip^i$. Then
\[{\binom{m}{n}\equiv\prod_i\binom{m_i}{n_i}\pmod p.}\]
\end{theorem}

In particular, for $p=2$, $\binom{m}{n}$ is odd if and only if $n_i\le m_i$ for all $i$, i.e., the binary representation of $n$ is bitwise dominated by that of $m$.

For self-orthogonality we need the stronger 2-adic valuation of binomial coefficients, given by Kummer's theorem \cite{Kummer}.

\begin{theorem}[Kummer] Let $p$ be a prime and $0\le n\le m$. The exponent of the highest power of $p$ dividing $\binom{m}{n}$ equals the number of carries when $n$ and $m-n$ are added in base $p$.
\end{theorem}

For $p=2$, this means that $\binom{m}{n}$ is even if and only if there is at least one carry in the binary addition $n+(m-n)=m$.
\subsection{Symmetric difference and support inclusion}
For vectors $u,v\in\mathbb F_2^k$, vector addition corresponds to symmetric difference of supports:
\[\operatorname{supp}(u)\triangle\operatorname{supp}(v)=(\operatorname{supp}(u)\cup\operatorname{supp}(v))\setminus(\operatorname{supp}(u)\cap\operatorname{supp}(v)).\]
In particular, $\operatorname{supp}(u+v)=\operatorname{supp}(u)\triangle\operatorname{supp}(v)$.

For defining-set codes, support containment is characterized coordinatewise:
	\(\operatorname{supp}(\bc_a)\subseteq\operatorname{supp}(\bc_b)\) if and only
	if there is no \(x\in D\) with \(a\cdot x=1\) and \(b\cdot x=0\).
	This observation will be used in the following lemma.

\begin{lemma}\label{lem:defining-set-support} For a defining-set code $\C_D$, a codeword $\bc=(c_v)_{v\in D}\in\C_D$ corresponds to $x\in\mathbb F_2^k$ via $c_v=\langle x,v\rangle$. Another codeword $\bc'$ (from $x'$) satisfies $\operatorname{supp}(\bc')\subseteq\operatorname{supp}(\bc)$ if and only if for every $v\in D$,  \[\langle x',v\rangle=1\implies \langle x,v\rangle=1.\]\end{lemma}

\subsection{Graph theory and cut spaces}
Let $G=(V(G),E(G))$ be a simple undirected graph. The cut associated with vertex $v\in V(G)$ is
\[\delta_G(v)=\{\{v,w\}\in E(G):w\in V(G)\},\]
viewed as a vector in $\mathbb F_2^{\#E(G)}$ indicating which edges are incident with $v$.

The cut space (or cut code) of $G$ is
\[\mathcal C_G=\langle 1_{\delta_G(v)}:v\in V(G)\rangle\subseteq\mathbb F_2^{\#E(G)}.\]

The binary vertex--edge incidence matrix gives the standard formula
\cite{DankelmannKeyRodrigues2013}
\[\dim(\mathcal C_G)=\#V-c,\] where $c$ is the number of connected components of $G$.

For a complete graph $K_r$, $\dim(\mathcal C_{K_r})=r-1$ and $\#E(K_r)=\binom r2=\frac{r(r-1)}{2}$.
For any two vertices $u,v\in V(K_r)$, the inner product $\langle1_{\delta_G(u)},1_{\delta_G(v)}\rangle$ counts, modulo $2$, the edges incident with both vertices:
\[\langle\delta_G(u),\delta_G(v)\rangle=
\begin{cases}
	\deg(u)\pmod 2, & u=v,\\
	1, & u\ne v\text{ and }uv\in E(K_r),\\
	0, & u\ne v\text{ and }uv\notin E(K_r).
\end{cases}\]

These identities specialize the known incidence-matrix self-orthogonality criterion recalled in Section~\ref{sec:graph-cut-space}.

\section{Minimality and refined parameters for single weight-layer codes}
\label{sec:a-sin-weight}

The codes in this section coincide with the complete constant-weight
layer codes studied in~\cite{ZhouLiTangDing2019}.  In simplicial-complex
language they satisfy
\[
D_t=\Delta_{\le t}\setminus\Delta_{\le t-1},
\]
an identity explicitly recorded by Wu and Lee~\cite{Wu-lee}.  Hence the
defining set and its dimension and Krawtchouk weight formula are known.
We recall those facts only because they are needed for the complete
minimality classification, the exact \(t=4\) analysis, and the refined
dual-distance condition proved below.

In this section, we assume $k\ge3$, $2\le t\le k-1$ and $V=\F_2^k$. Take
\[D_t=\{x\in V: \operatorname{wt}(x)=t\},\qquad n_t=\#D_t=\binom{k}{t}.\]
\begin{definition}\label{def:code-constru}Define the linear map
\[\phi_{D_t}:\mathbb{F}_{2}^{k}\longrightarrow \F_2^{n_t},\qquad a\longmapsto \bc_{a}=\big(\langle a,x\rangle\big)_{x\in D_t}.\]
The linear code $\C_{k,t}$ is defined as
\[\C_{k,t}:=\{\bc_a=(\langle a,x\rangle)_{x\in D_t}:\ a\in\F_2^k\}.\]\end{definition}
Clearly, the length of $\C_{k,t}$ is $n_t$.  We quote only the
previously known facts needed in the new arguments below.

\begin{lemma}\cite[Lemma~17]{ZhouLiTangDing2019} \label{lem:2.1}
Let \(k\ge3\) and \(1\le t\le k-1\). Then
\[
\langle D_t\rangle=
\begin{cases}
\mathbb F_2^k,&t\text{ odd},\\
\{x\in\mathbb F_2^k:\operatorname{wt}(x)\equiv0\pmod2\},&t\text{ even}.
\end{cases}
\]
Consequently,
\[
\ker\phi_{D_t}=
\begin{cases}\{0\},&t\text{ odd},\\
\langle {\one}\rangle,&t\text{ even},
\end{cases}
\qquad
\dim\C_{k,t}=
\begin{cases}k,&t\text{ odd},\\ k-1,&t\text{ even}.
\end{cases}
\]
\end{lemma}

The weight formula of Zhou et al.~\cite{ZhouLiTangDing2019} takes the
following form in our notation.
\begin{theorem}\label{thm:2.1}
If \(a\in\mathbb F_2^k\) has weight \(s\), then
\begin{equation}\label{eq:weight}
w_s=\sum_{\substack{0\le j\le t\\j\ \mathrm{odd}}}
\binom sj\binom{k-s}{t-j},
\end{equation}
or equivalently
\begin{equation}\label{eq:cut-weight}
w_s=\frac{n_t-K_t^{(k)}(s)}2,
\qquad
K_t^{(k)}(s)=\sum_{j=0}^t(-1)^j\binom sj\binom{k-s}{t-j}.
\end{equation}
In particular,
\begin{equation}\label{eq:nd}
\sum_{x\in D_t}(-1)^{\langle a,x\rangle}=K_t^{(k)}(s).
\end{equation}
The weight enumerator is
\[
A_{\C_{k,t}}(z)=
\begin{cases}
\displaystyle\sum_{s=0}^k\binom ks z^{w_s},&t\text{ odd},\\[6pt]
\displaystyle\frac12\sum_{s=0}^k\binom ks z^{w_s},&t\text{ even}.
\end{cases}
\]
\end{theorem}

\begin{corollary}\label{cor:cor1}
Let \(I_t=\{1,\ldots,k\}\) for odd \(t\), and
\(I_t=\{1,\ldots,k-1\}\) for even \(t\). Then
\[
d(\C_{k,t})=\frac12\left(n_t-\max_{s\in I_t}K_t^{(k)}(s)\right),
\qquad
w_{\max}(\C_{k,t})=\frac12\left(n_t-\min_{s\in I_t}K_t^{(k)}(s)\right).
\]
For odd \(t\), the latter formula includes \(w_{\max}=n_t\), attained
by \(\bc_{{\one}}\).
\end{corollary}

\begin{corollary}\label{cor:cor2}
For \(t=k-1\), the code is \([k,k,1]\) when \(k\) is even and is the
\([k,k-1,2]\) even-weight code when \(k\) is odd.
\end{corollary}

We next give an explicit example.
\begin{theorem}\label{thm:2.2}
 Let  $t=4$ and $k\geq6$. Suppose $s={\operatorname{wt}}(a)$ and
 \[r=s(k-s),\qquad y=|k-2s|.\]
 Then
 \begin{equation}\label{eq:case-4}
   K_4^{(k)}(s)
=\frac{y^4-(6k-8)y^2+3k(k-2)}{24},\ w_s
=\frac{(k^2-y^2)((k-2)(k-4)+y^2)}{48}.
 \end{equation}
 Moreover, the minimum distance
 \[
d(\C_{k,4})=
\begin{cases}
6,&k=6,\\
16,&k=7,\\
32,&k=8,\\
\binom{k-1}{3},&k\ge9.
\end{cases}
\]
\end{theorem}
\begin{proof}
 By~\eqref{eq:weight},
 \[
\begin{aligned}
w_s=s\binom{k-s}{3}+\binom{s}{3}(k-s)
=\frac{(k^2-y^2)((k-2)(k-4)+y^2)}{48}.
\end{aligned}
\]By~\eqref{eq:cut-weight},
\[
\begin{aligned}
K_4^{(k)}(s)
&=n_t-2w_s
=\binom{k}{4}-2s\binom{k-s}{3}-2(k-s)\binom{s}{3}\\
&=\frac{y^4-(6k-8)y^2+3k(k-2)}{24}.
\end{aligned}
\]


Since \(t=4\) is even, \(\bc_{a+{\one}}=\bc_a\), and hence
\(w_s=w_{k-s}\). It therefore suffices to consider
\(1\le s\le\lfloor k/2\rfloor\). On this range,
\(r=s(k-s)\) is strictly increasing. The function
\[
w(r):=w_s(r)=\frac{r(k^2-3k+4-2r)}{6}
\]
satisfies \(w''(r)=-2/3<0\), and hence is concave in \(r\).
Therefore, its minimum on the relevant \(r\)-interval is
attained at an endpoint, corresponding to
\[
s=1\qquad\text{or}\qquad
s=\left\lfloor\frac{k}{2}\right\rfloor.
\]

We next evaluate $w_s$ at the two endpoints
\begin{itemize}
  \item At $s=1$.\begin{equation}\label{eq:6.2}r(1)=k-1,\qquad w_1=w_1(k-1)=\binom{k-1}{3}=\frac{(k-1)(k-2)(k-3)}{6}. \end{equation}
  \item At $s=\lfloor k/2\rfloor$. Writing $\lfloor k/2\rfloor=k/2-\varepsilon$ with $\varepsilon\in\{0,1/2\}$, one has
     \begin{equation}\label{eq:6.21}r\!\left(\left\lfloor\tfrac{k}{2}\right\rfloor\right)
=\left\lfloor\tfrac{k}{2}\right\rfloor\!\left(k-\left\lfloor\tfrac{k}{2}\right\rfloor\right)
=\left\lfloor\tfrac{k^2}{4}\right\rfloor, \end{equation}and consequently
\begin{equation}\label{eq:6.22}w_{\lfloor k/2\rfloor}
  =\frac{\lfloor k^2/4\rfloor\,(k^2-3k+4-2\lfloor k^2/4\rfloor)}{6}. \end{equation}
\end{itemize}

Put
\[
R=\left\lfloor\frac{k^2}{4}\right\rfloor.
\]
Comparing ~\eqref{eq:6.2} and ~\eqref{eq:6.22}, we obtain
\[
\Delta(k):=w_1-w_{\lfloor k/2\rfloor}
=
\frac{(k-1-R)(k^2-5k+6-2R)}{6}.
\]

For \(k=6,7,8\), the factor \(k-1-R\) is negative, while
\(k^2-5k+6-2R\) equals \(-6,-4,-2\), respectively. Hence
\(\Delta(k)>0\), and the central endpoint gives the minimum:
\[
w_{\min}=
\begin{cases}
	6,&k=6,\\
	16,&k=7,\\
	32,&k=8.
\end{cases}
\]

For \(k\ge9\),
\[
k^2-5k+6-2R
\ge \frac{k^2}{2}-5k+6>0,
\]
whereas \(k-1-R<0\). Thus \(\Delta(k)<0\), so
\[
w_1<w_{\lfloor k/2\rfloor}.
\]
Therefore,
\[
w_{\min}=w_1=\binom{k-1}{3}
\qquad(k\ge9).
\]
This completes the proof.
\end{proof}

\subsection{Minimality of \texorpdfstring{$\C_{k,t}$}{C(k,t)}}

We next establish the exact condition under which the code $\C_{k,t}$ is minimal.

\newcommand{\inner}[2]{\langle #1,#2\rangle}
\begin{theorem}\label{thm:minimality}
	Let $k\ge3$ and $2\le t\le k-1$. The code $\C_{k,t}$ is minimal if and only if either $t$ is even and $2\le t\le k-2$, or $(k,t)=(3,2)$.
\end{theorem}
\begin{proof}
	We divide the proof into three parts.

	{\bf Case 1: $t$ is odd.}
	Assume $t$ is odd. Consider the all-one vector ${\one}\in\mathbb{F}_2^k$. For every $x\in D_t$, $\mathrm{wt}(x)=t$ is odd, so $\langle{\one},x\rangle=1$. Hence $\bc_{{\one}}=(1,1,\dots,1)$ and $\mathrm{wt}(\bc_{\one})=n_t$.
	
	Now take $a=e_1=(1,0,\ldots,0)$. By~\eqref{eq:weight},
	\[\mathrm{wt}(\bc_{e_1})=\binom{k-1}{t-1} < \binom{k}{t}=n_t,\]
	
	because $t\ge 2$. Thus $\mathrm{supp}(\bc_{e_1})\subset\mathrm{supp}(\bc_{\one})$, which means $\C_{k,t}$ is not minimal.
	\smallskip

{\bf	Case 2: $t$ is even and $2\le t\le k-2$.}
		Suppose, for a contradiction, that there exist distinct nonzero $\bc_a,\bc_b$ with
\[
{\operatorname{supp}}(\bc_a)\subseteq {\operatorname{supp}}(\bc_b).
\]
	
For $x\in D_t$, by Lemma~\ref{lem:defining-set-support}, we know
\[
\inner{a}{x}=1
\implies \inner{b}{x}=1.
\]Equivalently,
\begin{equation}\label{eq:2.1}
\inner{a}{x}(1+\inner{b}{x})=0
\quad\text{for every }x\in D_t.\end{equation}

Define the Boolean polynomial
\[
P(x)=\inner{a}{x}\,\big(1+\inner{b}{x}\big),\qquad x\in \F_2^k.
\]
It has degree at most two and by ~\eqref{eq:2.1}
\begin{equation}\label{eq:2.2}
P(x)=0\quad\text{for all }x\in D_t. \end{equation}

Write it in multilinear form as
\[
P(x)=\inner{a}{x}+\inner{a}{x}\inner{b}{x}
=\sum_{i=1}^k l_i x_i+\sum_{1\le i<j\le k}q_{ij}\,x_i x_j,
\]
where the coefficients are
\begin{equation}\label{eq:2.3} l_i=a_i(1+b_i),\qquad
q_{ij}=a_i b_j+a_j b_i\qquad (i\ne j). \end{equation}
Take four distinct indices $i,j,p,q$. Since $t\leq k-2$, we can choose a set $R$ of size $t-2$ disjoint from those four indices. Consider the four $t$-subsets
\[R \cup \{i,j\},\ R \cup \{i,p\},\ R \cup \{q,j\},\ R \cup \{q,p\}.\]
The polynomial $P$ vanishes on all four sets. Adding the four resulting equations over $\F_2$, all linear and common quadratic terms cancel. The remaining terms are
\begin{equation}\label{eq:9}
		q_{ij}+q_{ip}+q_{qj}+q_{qp}
	 = 0. \end{equation}

	For distinct \(i,j,\ell\), put
	\[
	T_{ij\ell}=q_{ij}+q_{i\ell}+q_{j\ell}.
	\]
	For four distinct indices \(i,j,p,\ell\), Eq.~\eqref{eq:9} gives
	\[
	T_{ij\ell}+T_{ip\ell}
	=q_{ij}+q_{ip}+q_{\ell j}+q_{\ell p}=0.
	\]
	Thus, two triangles sharing an edge have the same triangle sum.
	Since all triangles can be connected by successively replacing
	one vertex, their triangle sums have a common value
	\(\gamma\in\mathbb F_2\). Taking
	\[
	r_1=0,\qquad r_i=q_{1i}+\gamma\quad(i>1),
	\]
	gives
\begin{equation}\label{eq:10}q_{ij} = r_i + r_j + \gamma\quad (i\neq j). \end{equation}

For distinct \(i,j\), choose a \((t-1)\)-subset \(R\) disjoint from
\(i,j\). Comparing \(P\) on \(R\cup\{i\}\) and \(R\cup\{j\}\),
and using {\eqref{eq:10}}, gives
\[
0=l_i+l_j+\sum_{u\in R}(q_{iu}+q_{ju})
=l_i+l_j+(t-1)(r_i+r_j).
\]
Since \(t\) is even, this becomes \(l_i+r_i=l_j+r_j\). Thus
\(l_i+r_i\) is independent of \(i\), and {\eqref{eq:10}} yields
\[
q_{ij}=l_i+l_j+\gamma.
\]

Using~\eqref{eq:2.3}, the preceding relation becomes
\begin{equation}\label{eq:11}
  a_ib_j+a_jb_i=a_i(1+b_i)+a_j(1+b_j)+\gamma.
\end{equation}

Classify each coordinate $i$ according to its type $(a_i,b_i)\in\{00,01,10,11\}$. If $\gamma=1$, direct substitution into~\eqref{eq:11} shows that the types of any two coordinates must form either $\{00,10\}$ or $\{01,11\}$. No three coordinates can satisfy this pairwise condition. Since $k\ge4$, this is impossible. Therefore $\gamma=0$.

When $\gamma=0$, ~\eqref{eq:11} implies the following
\begin{itemize}
  \item among coordinates with $b_i=0$, all values $a_i$ are equal,
  \item among coordinates with $b_i=1$, all values $a_i$ are equal.
\end{itemize}
Thus there exist $\alpha,\beta\in\F_2$ such that
\[a_i=\alpha+\beta b_i\quad \text{ for every }\ i.\]
Equivalently,
\begin{equation}\label{eq:12}
 a=\alpha\one+\beta b.
\end{equation}
Because $t$ is even, the kernel is $\langle\one\rangle$. Hence ~\eqref{eq:12} implies
\[a+\langle\one\rangle\in\{\langle{\one}\rangle,b+\langle\one\rangle\}.\]
Therefore, the two parameters $a$ and $b$ represent either the zero codeword or the same codeword. This contradicts the assumption that $\bc_a$ and $\bc_b$ are distinct nonzero codewords.

Consequently, for even $t$ with $2\leq t\leq k-2$, no proper support inclusion exists. Hence $\C_{k,t}$ is minimal.

\smallskip

{\bf Case 3: $t=k-1$.}
If $t=k-1$ is even, then $k$ is odd. The set $D_{k-1}$ consists of all vectors of weight
$k-1$, i.e.,
\[\one+e_i,\quad 1\leq i\le k.\] Hence
\[
(\bc_a)_i=\inner{a}{\one+e_i}=\sum_{j=1}^ka_j+a_i
\]
By Corollary~\ref{cor:cor2}, the resulting code is the length-$k$ even-weight code.

If $k=3$, this code has parameters $[3,2,2]$, and every nonzero codeword has weight 2. Therefore, no nonzero support can properly contain another, so $\C_{3,2}$ is minimal.

If $k\geq5$, the even-weight code contains the two codewords
\[u=e_1+e_2,\qquad v=e_1+e_2+e_3+e_4.\] These supports satisfy \[{\operatorname{supp}}(u)=\{1,2\}\subset\{1,2,3,4\}={\operatorname{supp}}(v),\]
hence $\C_{k,k-1}$ is not minimal for odd $k\ge5$. Combining all cases, we obtain the result.
\end{proof}
\begin{remark}
The case \(t=2\), including its weight distribution, minimality,
and Ashikhmin--Barg comparison, is due to Zhang et
al.~\cite{ZhangYanWei2019}, it is therefore not repeated here.  It is
included in Theorem~\ref{thm:minimality} only to make the new
all-layer classification complete.
\end{remark}

To illustrate the new
threshold, Theorems~\ref{thm:2.1} and~\ref{thm:2.2} give
\begin{align*}
A_{\C_{16,4}}(z) = 1&{}+16z^{455}+120z^{728}+560z^{871}
 +6435z^{896}+11440z^{903}\\
&{}+8008z^{920}+1820z^{928}+4368z^{935}.
\end{align*}
Thus \(455/935<1/2\).  We next prove that this strict
Ashikhmin--Barg violation holds for every \(k\ge16\).

By~\eqref{eq:case-4}, completing the square in the expression for $w_s$ gives
\[w_s = \frac{(k^2 - 3k + 4)^2 - (y^2 - (3k-4))^2}{48}.\]
Define
\[
\mathcal Y_k=
\mathop{\rm argmin}_{\substack{0\le y\le k-2\\y\equiv k\pmod2}}
\left|y^2-(3k-4)\right|.
\]
Then we have \[
w_{\max}
=
\frac{(k^2-3k+4)^2-
	\displaystyle\min_{\substack{0\le y\le k-2\\y\equiv k\pmod2}}
	(y^2-(3k-4))^2}{48},
\] and the maximum is attained exactly at
\[
s=\frac{k\pm y}{2},\qquad y\in\mathcal Y_k.
\]

Substituting $k=6,7,\ldots,16$ yields the following table.
\begin{table}[ht]
	\centering
	\caption{Numerical values of $d$ and $w_{\max}$ for $6 \le k \le 16$}
	\label{tab:values}
	\begin{tabular}{c|ccccccccccc}
		$k$ & 6 & 7 & 8 & 9 & 10 & 11 & 12 & 13 & 14 & 15 & 16 \\ \hline
		$d$ & 6 & 16 & 32 & 56 & 84 & 120 & 165 & 220 & 286 & 364 & 455 \\
		$w_{\max}$ & 10 & 20 & 40 & 70 & 112 & 176 & 261 & 372 & 520 &704 & 935 \\
	\end{tabular}
\end{table}
For $6 \le k \le 15$, Table~\ref{tab:values} directly gives $2d > w_{\max}$. Since $k\ge16$ implies
$0<\sqrt{3k-4}<k-2$, the integer $y$ of the same parity as $k$
 closest to $\sqrt{3k-4}$ lies in $[0,\ldots,k-2]$ and minimizes $|y^2-(3k-4)|$. Set
\[\delta = |y^2 - (3k-4)|.\]

Since $x^2\equiv2\pmod3$ has no integer solutions, $\sqrt{3k-4}$ is not an integer. Hence
\begin{equation}\label{eq:delta}|y-\sqrt{3k-4}|\le1\implies \delta = |y - \sqrt{3k-4}| \cdot |y + \sqrt{3k-4}|\le1\cdot (2\sqrt{3k-4}+1).\end{equation}

For $k\ge6$, we have $4\sqrt{3k-4}\le4k+15$. Therefore, by~\eqref{eq:delta},
\begin{align*}\delta^2 &\le (2\sqrt{3k-4} + 1)^2 = 4(3k-4) + 4\sqrt{3k-4} + 1 = 12k - 15 + 4\sqrt{3k-4}\\
	&\le 12k-15+4k+15=16k.\end{align*}
We further have
\begin{equation}\label{eq:de1}48w_{\max} = (k^2 - 3k + 4)^2 - \delta^2\ge (k^2 - 3k + 4)^2-16k. \end{equation}
Notice that \begin{equation}\label{eq:de2}48(2d) = 96\binom{k-1}{3} = 16(k-1)(k-2)(k-3).\end{equation}
Consequently,
\begin{align*} 48(w_{\max} - 2d) &\ge (k^2 - 3k + 4)^2 - 16k - 16(k-1)(k-2)(k-3) \\ &= k^4 - 22k^3 + 113k^2 - 216k + 112. \end{align*}

The polynomial on the right equals $1008$ at $k=16$. Its first derivative is positive at $k=16$, and its second and third derivatives are positive for $k\ge16$, so it is strictly increasing on this range. Therefore, we obtain the following result.
\begin{proposition}
		For every \(k\ge16\), \(\mathcal C_{k,4}\) is minimal and satisfies
		\[
		\frac{w_{\min}}{w_{\max}}<\frac12.
		\]
\end{proposition}
\subsection{Dual distance and self-orthogonality}

Let
\[
{G_{\mathrm{nat}}}=(x)_{x\in D_t}\in\mathbb F_2^{k\times n_t}
\]
be the natural generator matrix.  For every \(t\)-subset
\(A\subseteq[k]\), write
\[
x_A=1_A\in\mathbb F_2^k
\]
for the column of \({G_{\mathrm{nat}}}\) indexed by \(A\), and let
\(\varepsilon_A\in\mathbb F_2^{n_t}\) be the corresponding coordinate
unit vector.  Thus
\[
\C_{k,t}^{\perp}
=\{v\in\mathbb F_2^{n_t}:{G_{\mathrm{nat}}}v^T=0\}.
\]

\begin{lemma}\label{lem:columns}
Every column \(x_A\) is nonzero, and distinct \(t\)-subsets index
distinct columns.  Consequently,
\[
A_1(\C_{k,t}^{\perp})=A_2(\C_{k,t}^{\perp})=0.
\]
\end{lemma}
\begin{proof}
The columns are precisely the distinct nonzero incidence vectors
\(x_A=1_A\) with \(|A|=t\).
\end{proof}

\begin{proposition}\label{prop:dual-3}
Let \(A,B,C\subseteq[k]\) be three pairwise distinct \(t\)-subsets and
put
\[
v_{A,B,C}=\varepsilon_A+\varepsilon_B+\varepsilon_C
\in\mathbb F_2^{n_t}.
\]
Then
\[
v_{A,B,C}\in\C_{k,t}^{\perp}
\iff x_A+x_B+x_C=0\quad\text{in }\mathbb F_2^k
\iff C=A\mathbin\triangle B\ \text{ and }\ |A\cap B|=\frac t2.
\]
\end{proposition}
\begin{proof}
Since \({G_{\mathrm{nat}}}\varepsilon_A^T=x_A\),
\[
{G_{\mathrm{nat}}}v_{A,B,C}^T=x_A+x_B+x_C.
\]
Hence \(v_{A,B,C}\in\C_{k,t}^{\perp}\) exactly when
\(x_C=x_A+x_B\), which is equivalent to \(C=A\triangle B\).  Because
\(|A|=|B|=|C|=t\),
\begin{equation}\label{eq:2.6}
t=|C|=|A\triangle B|=2t-2|A\cap B|,
\end{equation}
and thus \(|A\cap B|=t/2\).  The converse follows by reversing
the same argument.
\end{proof}

\begin{theorem}
For \(k\ge4\) and \(2\le t\le k-2\),
\[
d(\C_{k,t}^{\perp})=
\begin{cases}
3,&t\text{ is even and }k\ge3t/2,\\
4,&t\text{ is odd, or }t\text{ is even and }k<3t/2.
\end{cases}
\]
\end{theorem}
\begin{proof}
By Proposition~\ref{prop:dual-3}, a three-column relation exists exactly
when two \(t\)-subsets \(A,B\) can satisfy \(|A\cap B|=t/2\).  This is
possible exactly when \(t\) is even and
\[
k-t\ge\frac t2,
\]
or equivalently \(k\ge3t/2\).  Together with
Lemma~\ref{lem:columns}, this proves the first case and gives the lower
bound four in the remaining cases.

It remains to construct a relation of weight four.  Choose
\(R\subseteq[k]\) with \(|R|=t-2\) and four distinct elements
\(i,j,p,q\notin R\), and set
\begin{equation}\label{eq:2.8}
A=R\cup\{i,j\},\quad B=R\cup\{i,p\},\quad
C=R\cup\{q,j\},\quad D=R\cup\{q,p\}.
\end{equation}
These four \(t\)-subsets are distinct and their incidence columns satisfy
\begin{equation}\label{eq:2.9}
x_A+x_B+x_C+x_D=0\qquad\text{in }\mathbb F_2^k.
\end{equation}
Therefore
\[
v=\varepsilon_A+\varepsilon_B+\varepsilon_C+\varepsilon_D
\]
is a codeword of weight four in \(\C_{k,t}^{\perp}\).  This proves the
second case.
\end{proof}

\begin{remark}
The condition \(k\ge3t/2\) is essential: a previously stated
even-\(t\) criterion in~\cite[Theorem~25]{ZhouLiTangDing2019} does not
include this feasibility restriction. For example, when \((k,t)=(8,6)\),
three weight-six columns cannot sum to zero, because two such supports
would have to meet in three points and their union would have size
nine. The four-column construction in the preceding proof therefore gives
\(d(\C_{8,6}^{\perp})=4\).
\end{remark}

In particular, for $t=4$ and $k\ge6$,
\[{d(\C_{k,4}^\perp) = 3.}\]

	We use the convention \(d(\{0\})=+\infty\). For the boundary $t=k-1$, we have
\begin{itemize}
	\item If $k$ is even, then $\C_{k,k-1}^\perp = \{0\}$, so there are no nonzero dual codewords.
    \item If $k$  is odd, then \(\C_{k,k-1}\) has parameters $[k,k-1,2]$ and $\C_{k,k-1}^\perp = \langle {\one} \rangle$, and thus
    \[{d(\C_{k,k-1}^\perp) = k.}\]
\end{itemize}

The following previously known row-intersection data are recalled only
to state the doubly-even refinement used below.
\begin{lemma}\cite[Lemma~18 and Theorem~21]{ZhouLiTangDing2019}\label{lem:4}
Let \(G_{\rm nat}\) be the natural generator matrix of \(\C_{k,t}\).
Every row has weight
\begin{equation}\label{eq:5.2}R=\binom{k-1}{t-1},\end{equation}
and two distinct rows have common support of size
\begin{equation}\label{eq:5.3}P=\binom{k-2}{t-2}.\end{equation}
Consequently, \(\C_{k,t}\) is even if and only if \(R\equiv0\pmod2\), and it is
self-orthogonal if and only if \(R\equiv P\equiv0\pmod2\).
\end{lemma}

\begin{proposition}
The code \(\C_{k,t}\) is doubly even if and only if
\[
R\equiv0\pmod4,
\qquad P\equiv0\pmod2.
\]\end{proposition}
\begin{proof}

Indeed, necessity follows because each generating row has weight
\(R\), while
\(\operatorname{wt}(r_i+r_j)=2R-2P\) for distinct generating rows.

For sufficiency, if $J$ is any set of row indices, then
\[
\operatorname{wt}\!\left(\sum_{i\in J}r_i\right)
\equiv \sum_{i\in J}\operatorname{wt}(r_i)
-2\sum_{\substack{i<j\\ i,j\in J}}
\bigl|\operatorname{supp}(r_i)\cap\operatorname{supp}(r_j)\bigr|
\pmod 4.
\]
Thus $R\equiv0\pmod4$ and $P\equiv0\pmod2$ are also sufficient for
every linear combination of the rows to have weight divisible by four.
\end{proof}

The parity of binomial coefficients is determined by Lucas's theorem:
	\begin{equation}\label{eq:lucas}\binom{n}{r} \equiv 1 \pmod 2 \iff r \,\&\, (n-r) = 0,\end{equation}
where $\&$ denotes the bitwise AND.
	
Applying this to $\binom{k-1}{t-1}$ and $\binom{k-2}{t-2}$ yields
\begin{align*}\binom{k-1}{t-1} \text{ is even}& \iff (t-1) \,\&\, (k-t) \neq 0, \\
	\binom{k-2}{t-2} \text{ is even} & \iff (t-2) \,\&\, (k-t) \neq 0, \end{align*}
	Now specialize to $t=4$ (with $k\ge 6$). Then
	\[R = \binom{k-1}{3}, \qquad P = \binom{k-2}{2}.\]
	By ~\eqref{eq:lucas}, $\binom{k-1}{3}$ is odd exactly when $(3) \,\&\, (k-4) = 0$, which occurs precisely for $k \equiv 0 \pmod 4$. Thus it is even iff $k \not\equiv 0 \pmod 4$, i.e.,
	\[\C_{k,4} \text{ is even} \iff k \not\equiv 0 \pmod 4.\]

	For self-orthogonality, we need both $R$ even and $P$ even. The parity criterion~\eqref{eq:lucas} gives that $\binom{k-2}{2}$ is even iff \[(2) \,\&\, (k-4) \neq 0,\] which holds for $k \equiv 2,3 \pmod 4$. Intersecting with $R$ even ($k \not\equiv 0 \pmod 4$) yields exactly $k \equiv 2,3 \pmod 4$, i.e.,
\[\C_{k,4} \text{ is self-orthogonal} \iff k \equiv 2,3 \pmod 4. \]
	
	For doubly-even, we require $R \equiv 0 \pmod 4$ and $P$ even. Within the self-orthogonal residue classes $k \equiv 2,3,6,7 \pmod 8$, one checks that $\binom{k-1}{3}$ is divisible by 4 precisely when $k \equiv 2,3,7 \pmod 8$ (excluding $k \equiv 6 \pmod 8$). This yields
	\[\C_{k,4} \text{ is doubly-even} \iff k \equiv 2,3,7 \pmod 8. \]

\section{Minimality classification for arbitrary unions of weight layers}
\label{sec:union-weight-layer}
In this section, let \(k\ge4\) and
\(\varnothing\neq S\subseteq\{1,\ldots,k-1\}\).
The relation with simplicial-complex constructions can be stated
exactly.  Initial unions satisfy
\[
D_{\{1,\ldots,r\}}:=D_1\cup D_2\cup\cdots\cup D_r=\Delta_{\le r}\setminus\{0\},
\]
as in the defining sets considered in
\cite{Wu-lee, ZhouLiTangDing2019}.  More generally, a consecutive block
\(H=\{a,a+1,\ldots,b\}\) gives
\[
D_H=\Delta_{\le b}\setminus\Delta_{\le a-1},
\]
which is a difference of two simplicial complexes of the type allowed
in~\cite{Wu-lee}.  If the maximal consecutive-interval decomposition is
\[
H=[a_1,b_1]\uplus\cdots
 \uplus[a_m,b_m],
\]
where $[a_i,b_i]=\{a_i,a_i+1,\ldots,b_i\}$. Then
\[
D_H=\bigcup_{i=1}^{m}
\bigl(\Delta_{\le b_i}\setminus\Delta_{\le a_i-1}\bigr).
\]
Hence, arbitrary layer unions remain within the broader framework of set
operations on simplicial complexes.  Incidence codes of hypergraphs
provide another ambient framework~\cite{MallikYildiz2024}.  Our aim is
the uniform treatment of every layer-index set \(S\), especially the
if-and-only-if minimality classification and the simultaneous parameter
criteria below.
We take
\[
D_S=\bigcup_{t\in S}D_t,
\]
and denote the resulting code by \(\C_{k,S}\).
Set \[
S_{\mathrm{odd}}=\{t\in S:t\equiv1\pmod2\},\qquad
S_{\mathrm{even}}=\{t\in S:t\equiv0\pmod2\}.
\]Then
\begin{equation}\label{eq:gen-defi}
	D_S=(\bigcup_{t\in S_{\mathrm{odd}}}D_t)\bigcup(\bigcup_{t\in S_{\mathrm{even}}}D_t)
\end{equation}
\begin{theorem}
For the defining set $D_S$ in~\eqref{eq:gen-defi},
  \[
\ker\phi_{D_S}=
\begin{cases}
\{0\},&S_{\rm odd}\ne\varnothing,\\
\langle\one\rangle,&S_{\rm odd}=\varnothing.
\end{cases}\implies
\dim\C_{k,S}=
\begin{cases}
k,&S_{\rm odd}\ne\varnothing,\\
k-1,&S_{\rm odd}=\varnothing.
\end{cases}
\]
Moreover, if $\operatorname{wt}(a)=s$, then the weight of $\bc_a$ is
\[
w_s=\frac12\sum_{t\in S}
\left(\binom kt-K_t^{(k)}(s)\right),\qquad A_{\C_{k,S}}(z)=
\frac1{\#\ker\phi_{D_S}}
\sum_{s=0}^k\binom ks z^{w_s}
\]
\end{theorem}
\begin{proof}
By Lemma~\ref{lem:2.1}, on each individual layer $D_t$, $\ker\phi_{D_t}=\{0\}$ for odd $t$ and $=\langle\one\rangle$ for even $t$.

Since $\F_2^k$ decomposes as the disjoint union of its Hamming
layers,
\[
\mathbb{F}_2^k=\biguplus_{s=0}^{k}V_s,\qquad V_s=\{a\in\mathbb{F}_2^k:\mathrm{wt}(a)=s\},
\] the evaluation map is $\phi_{D_S}=(\phi_{D_t})_{t\in S}$, and therefore
\[\ker\phi_{D_S}=\bigcap_{t\in S}\ker\phi_{D_t}.\]

If $S$ contains an odd $t$, then $\ker\phi_{D_t}=\{0\}$ for that layer, forcing $\ker\phi_{D_S}=\{0\}$. If all $t\in S$ are even, each $\ker\phi_{D_t}=\langle\one\rangle$, so their intersection is again $\langle\one\rangle$. The dimension of $\C_{k,S}$ is directly obtained.

For $a\in\mathbb{F}_2^k$ with $\mathrm{wt}(a)=s$, write $w_s^{(t)}$ for the weight of the restriction of $\bc_a=\phi_{D_S}(a)$ to $D_t$. Then
\[
w_s^{(t)}
=\sum_{\substack{0\le j\le t\\j\ {\rm odd}}}
\binom sj\binom{k-s}{t-j}
=\frac12\left(\binom kt-K_t^{(k)}(s)\right).
\]
The total Hamming weight of $\bc_a$ is the sum over all layers
\[w(\bc_a)=\sum_{t\in S}w_s^{(t)}=\sum_{t\in S}\frac12\left(\binom kt-K_t^{(k)}(s)\right).\]

	The weight enumerator of \(\C_{k,S}\) is
\[
	A_{\C_{k,S}}(z)
=
\frac1{\#\ker\phi_{D_S}}
\sum_{s=0}^k\binom ks z^{w_s}.\qedhere
\]
\end{proof}

For $a\in\F_2^k$,
\[{\operatorname{wt}}(a)=s\implies {\operatorname{wt}}(\bar{a}):={\operatorname{wt}}(a+\one)=k-s.\] Note that
\begin{align*}
K_t^{(k)}(k-s)&=\sum_{j=0}^{t}(-1)^j\binom{k-s}{j}\binom{s}{t-j}=(-1)^t\sum_{j=0}^{t}(-1)^{t-j}\binom{k-s}{j}\binom{s}{t-j}\\
&=(-1)^tK_t^{(k)}(s).
\end{align*}
For any $X\in D_t$, $\inner{a+\one}{X}=\inner{a}{X}+t\pmod2$, then
\[
w_{k-s}^{(t)}=
\begin{cases}
w_s^{(t)},& 2\mid t,\\
n_t-w_s^{(t)},&2\nmid t,
\end{cases}
\]and the total weight is
\[
w_{k-s}
=
\sum_{t\in S_{\rm even}}w_s^{(t)}
+\sum_{t\in S_{\rm odd}}\bigl(n_t-w_s^{(t)}\bigr).
\]By the preceding weight formula and the definition of \(\ker\phi_{D_S}\),
\begin{equation}\label{eq:gen-code1}
  d(\C_{k,S})=w_{\min}(\C_{k,S})=\min_{s\in I_S}w_s,\qquad
w_{\max}(\C_{k,S})=\max_{s\in I_S}w_s,
\end{equation}
where \[
I_S=
\begin{cases}
	\{1,\ldots,k\},&S_{\rm odd}\ne\varnothing,\\
	\{1,\ldots,k-1\},&S_{\rm odd}=\varnothing.
\end{cases}
\]
\subsection{Minimality of \texorpdfstring{$\C_{k,S}$}{C(k,S)}}
For $(a,b)\in\F_2^k\times\F_2^k$, set
\begin{align*}
  p=\#\{i:(a_i,b_i)=(0,0)\},\quad & q=\#\{i:(a_i,b_i)=(0,1)\}, \\
  r=\#\{i:(a_i,b_i)=(1,0)\},\quad & h=\#\{i:(a_i,b_i)=(1,1)\}.
\end{align*}
Then $p+q+r+h=k$.
\begin{lemma}
  Take $M_t(a,b):=\#\{X\in D_t:\ \langle a,X\rangle=1, \langle b,X\rangle=0\}$.
  Then
  \begin{equation}\label{eq:mabb}
    M_t(a,b)=
\sum_{\substack{i+j+\ell+m=t\\
\ell+m\equiv1\ \pmod2\\
j+m\equiv0\ \pmod2}}
\binom pi\binom qj\binom r\ell\binom hm.
  \end{equation}
\end{lemma}
\begin{proof}
  A $t$-subset $X\in D_t$ is formed by selecting $i$ positions from the $(0,0)$ type, $j$ from $(0,1)$, $\ell$ from $(1,0)$, and $m$ from $(1,1)$, subject to $i+j+\ell+m=t$. The inner products modulo 2 are
\[{\langle a, X \rangle \equiv \ell + m \pmod 2, \qquad \langle b, X \rangle \equiv j + m \pmod 2.}\]

The conditions $\langle a, X \rangle = 1$ and $\langle b, X \rangle = 0$ thus translate exactly to the parity constraints $\ell+m \equiv 1 \pmod 2$ and $j+m \equiv 0 \pmod 2$. Summing the binomial choices $\binom{p}{i}\binom{q}{j}\binom{r}{\ell}\binom{h}{m}$ over all valid $(i,j,\ell,m)$ yields ~\eqref{eq:mabb}.
\end{proof}

Since
\[{1_{\{X\in D_t: \inner{a}{X} = 1, \inner{b}{X}= 0\}} = \frac{1}{4} \left( 1 - (-1)^{\langle a, X \rangle} \right) \left( 1 + (-1)^{\langle b, X \rangle} \right).}\]
Expanding this product gives four terms
\[{\frac{1}{4} \left[ 1 - (-1)^{\langle a, X \rangle} + (-1)^{\langle b, X \rangle} - (-1)^{\langle a+b, X \rangle} \right].}\]
By~\eqref{eq:nd}, $\sum_{X \in D_t} (-1)^{\langle u, X \rangle} = K_t^{(k)}(\mathrm{wt}(u))$. Therefore,\begin{align*}M_t(a,b)&=\sum_{X\in D_t}\frac{1}{4} \left[ 1 - (-1)^{\langle a, X \rangle} + (-1)^{\langle b, X \rangle} - (-1)^{\langle a+b, X \rangle} \right]\\
&=\frac{1}{4}\bigl(n_t-K_t^{(k)}({\operatorname{wt}}(a))+K_t^{(k)}({\operatorname{wt}}(b))-K_t^{(k)}({\operatorname{wt}}(a+b))\bigr).\end{align*}
Set $s_a=r+h$, $s_b=q+h$, and $s_{a+b}=q+r$. Then $s_a={\operatorname{wt}}(a)$, $s_b={\operatorname{wt}}(b)$, and $s_{a+b}={\operatorname{wt}}(a+b)$, so
\begin{equation}\label{eq:mab}
 M_t(a,b)=\frac{1}{4}\bigl(n_t-K_t^{(k)}(s_a)+K_t^{(k)}(s_b)-K_t^{(k)}(s_{a+b})\bigr)
\end{equation}
\begin{proposition}\label{pro:minma}
  The support of the codeword $\bc_a$ is contained in that of $\bc_b$ if and only if
\begin{equation}\label{eq:minmab}\sum_{t \in S} M_t(a,b) = 0.\end{equation}
\end{proposition}
\begin{proof}
 Recall that $\bc_a = (\langle a, X \rangle)_{X \in D_S}$ and $\bc_b = (\langle b, X \rangle)_{X \in D_S}$ with $D_S= \bigcup_{t \in S} D_t$.

The condition $\mathrm{supp}(\bc_a) \subseteq \mathrm{supp}(\bc_b)$ means that, for every $X \in D_S$, $\langle a,X\rangle=1$ implies $\langle b,X\rangle=1$. Equivalently, there is no $X\in D_S$ such that
\[\langle a,X\rangle=1,\qquad \langle b,X\rangle=0.\]
Since $M_t(a,b)$ counts exactly such vectors in $D_t$, their total number in $D_S$ is \[\sum_{t\in S}M_t(a,b).\] Hence
${\operatorname{supp}}(\bc_a)\subseteq{\operatorname{supp}}(\bc_b)$ holds if and only if $\sum_{t\in S}M_t(a,b)=0$.
\end{proof}
\begin{remark}
  Proposition~\ref{pro:minma} means that the code is minimal if and only if for all ordered pairs of parameters representing distinct nonzero codewords,
\[{\sum_{t \in S} M_t(a,b) > 0.}\]
\end{remark}
\begin{proposition}\label{pro:odd}
  Let $2\nmid t$ and $3 \le t \le k-2$.  If
\begin{equation}\label{eq:3.16}\mathrm{supp}(\bc_a|_{D_t}) \subseteq \mathrm{supp}(\bc_b|_{D_t})\end{equation}
then
\begin{equation}\label{eq:3.17}
a = 0, \quad \text{or} \quad a = b, \quad \text{or} \quad b = \one.\end{equation}
Conversely, each case in~\eqref{eq:3.17} implies~\eqref{eq:3.16}.
\end{proposition}
\begin{proof}
  Eq.~\eqref{eq:3.16} is equivalent to the quadratic Boolean polynomial
\[P(x) = \langle a, x\rangle(1 + \langle b , x\rangle)\]
vanishing on $D_t$. Expanding $P(x)$ yields
\[P(x) = \sum_i l_i x_i + \sum_{i<j} q_{ij} x_i x_j,\]
where
\begin{equation}\label{eq:3.18}l_i = a_i(1+b_i), \qquad q_{ij} = a_i b_j + a_j b_i.\end{equation}

Take four distinct coordinates $i,j,u,v$ and a subset $R$ of size $t-2$ disjoint from those four indices. Summing $P$ over the four $t$-subsets $R \cup \{i,j\}$, $R \cup \{i,u\}$, $R \cup \{v,j\}$, $R \cup \{v,u\}$ gives
\[q_{ij} + q_{iu} + q_{vj} + q_{vu} = 0.\]

	The same triangle-sum argument used after ~\eqref{eq:9} in the proof of
	Theorem~\ref{thm:minimality} shows that there exist
	\(\rho_1,\dots,\rho_k,\gamma\in\mathbb{F}_2\) such that
\begin{equation}\label{eq:3.20}q_{ij} = \rho_i + \rho_j + \gamma. \end{equation}

Now take distinct $i,j$ and a subset $R$ of size $t-1$ disjoint from $i,j$. Comparing $P(R \cup \{i\})$ and $P(R \cup \{j\})$, the difference of quadratic terms is \[\sum_{u \in R} (q_{iu} + q_{ju}) = (t-1)(\rho_i + \rho_j) = 0,\] since $t-1$ is even. Thus $l_i = l_j$ for all $i,j$.

If $l_i=1$ for all $i$, then ~\eqref{eq:3.18} forces $(a_i, b_i) = (1,0)$ for all $i$, i.e., $a={\one}, b=0$. But on the odd layer $D_t$, $\inner{a}{X} = 1$ and $\inner{b}{X} = 0$ for all $x \in D_t$, contradicting $P|_{D_t} = 0$. Therefore $l_i = 0$, meaning the coordinate type $(1,0)$ does not appear. The remaining types are only $(0,0)$, $(0,1)$, $(1,1)$.

If all three types appear, we first show $\gamma = 0$ in~\eqref{eq:3.20}. For two coordinates $i,j$ of the same type, take a coordinate $u$ of another type. By~\eqref{eq:3.18}, $q_{iu} = q_{ju}$, and~\eqref{eq:3.20} gives \begin{equation}\label{eq:3.21}\rho_i = \rho_j.\end{equation} Since $3\le t\le k-2$, we have $k\ge5$. If all three types are nonempty, the pigeonhole principle therefore gives two distinct coordinates $i,j$ of the same type. For this pair, $q_{ij}=0$ by~\eqref{eq:3.18}, and~\eqref{eq:3.21} yields $0=q_{ij}=\rho_i+\rho_j+\gamma=\gamma$.

Now $\rho_i$ is constant within each type, denoted $\rho_{00}, \rho_{01}, \rho_{11}$. Cross-type $q_{ij}$ values are
\[{q_{00,01} = 0, \qquad q_{00,11} = 0, \qquad q_{01,11} = 1.}\]

With $\gamma=0$, ~\eqref{eq:3.20}  gives $\rho_{00}=\rho_{01}$, $\rho_{00}=\rho_{11}$, and $\rho_{01}+\rho_{11}=1$, a contradiction. Hence at most two types appear:
\begin{itemize}
\item If only types $(0,0)$ and $(0,1)$ occur, then $a=0$; if only $(0,0)$ and $(1,1)$ occur, then $a=b$; and if only $(0,1)$ and $(1,1)$ occur, then $b=\one$.
\item The single-type cases are included in these conclusions, direct substitution verifies the converse in each case.\end{itemize}Therefore the proof is finished.
\end{proof}
We now present several complete classifications of minimality.
\begin{theorem}\label{thm:even}
	Suppose $k\ge4$, $S=S_{\mathrm{even}}$, and $S$ contains an
	even integer $t_0$ with $2\le t_0\le k-2$. Then the code
	$\C_{k,S}$ is minimal.
\end{theorem}
\begin{proof}
	Suppose two nonzero codewords satisfy
	\begin{equation}\label{eq:7}\mathrm{supp}(\bc_a) \subseteq \mathrm{supp}(\bc_b).\end{equation}
	
	Since the kernel of the union set code is $\langle \one \rangle$, we have $a, b \notin \langle \one \rangle$. For any $t \in S$, the kernel of the single-layer parameter is also $\langle \one \rangle$, hence
	\[{\bc_a|_{D_t} \neq 0, \qquad \bc_b|_{D_t} \neq 0.}\]
	
Choose the internal even layer $t_0$ from the hypothesis. Restricting ~\eqref{eq:7} to $D_{t_0}$ yields
\[\mathrm{supp}(\bc_a|_{D_{t_0}}) \subseteq \mathrm{supp}(\bc_b|_{D_{t_0}}).\]
	
	By Theorem~\ref{thm:minimality}, $\C_{k,t_0}$ is minimal. Therefore, the two restricted codewords must be equal, which implies
	\[a+b\in\ker\phi_{D_{t_0}}=\langle\one\rangle.\]
	Since every layer indexed by $S$ is even, adding $\one$ to a
	parameter does not change its evaluation on any $D_t$, $t\in S$.
	Consequently, $\bc_a = \bc_b$ holds on the entire $D_S$. This excludes any proper support containment between distinct nonzero codewords, hence $\C_{k,S}$ is minimal.
\end{proof}
\begin{theorem}\label{thm:classfic}
	Let $k\ge4$ and $\varnothing\neq S\subseteq\{1,2,\ldots,k-1\}$. The code $\C_{k,S}$ is minimal if and only if one of the following two conditions holds:
	\begin{enumerate}[label=\textup{(\arabic*)}]
		\item  $S$ contains some even integer $t_0$, $2\le t_0\le k-2$.
		\item  $k$ is odd, $k-1 \in S$, and $S \cap \{3,5,\dots,k-2\} \neq \varnothing$.
	\end{enumerate}
\end{theorem}
\begin{proof}
	(1)
	Choose an even \(t_0\in S\) with
	\(2\le t_0\le k-2\). If \(S\) contains only even integers,
	Theorem~\ref{thm:even} applies. Otherwise \(S\) contains an odd
	layer and the union evaluation map is injective. Let
	\(\bc_a,\bc_b\) be nonzero codewords satisfying
	\begin{equation}\label{eq:s}
		{\operatorname{supp}}(\bc_a)\subseteq{\operatorname{supp}}(\bc_b).
	\end{equation}
Suppose that
	\(a,b\notin
	\langle{\one}\rangle\).
	Then both restricted codewords on \(D_{t_0}\) are nonzero.
	Restricting~\eqref{eq:s} to \(D_{t_0}\) and applying the minimality of
	\(\C_{k,t_0}\) from Theorem~\ref{thm:minimality} gives
\begingroup
 \[a=b\ \text{ or }\ a=b+{\one}.\] The latter yields
\[\bc_a|_{D_t}=\bc_b|_{D_t}+\mathbf{1}_{D_t}\]
 on every odd layer $D_t$. If two binary vectors are complementary and
 the former's support is contained in the latter's, the former must be
 zero. Hence
 \[\bc_a|_{D_t}=\mathbf{0}_{D_t}\implies a\in\ker\phi_{D_t}.\]
\endgroup
Since the parameter map on every odd layer is injective, this forces
\(a=\mathbf{0}\), contradicting the assumption that
\(\bc_{\boldsymbol{a}}\) is nonzero.

If $a=\one$, it is the all-one codeword on odd layer $D_t$, so ~\eqref{eq:s} forces $b={\one}$. If $b=\one$, it vanishes on even internal layers, so ~\eqref{eq:s} forces $a \in \langle {\one} \rangle$, leaving only $a=b={\one}$. Thus the assertion (1) is sufficient.
\smallskip

(2) Now let $k$ be odd and choose an odd $r\in S$ with $3\le r\le k-2$, while $k-1\in S$. If~\eqref{eq:s} holds, restrict it to $D_r$. By Proposition~\ref{pro:odd}, the only possibility for distinct nonzero parameters is $b={\one}$. But $\bc_{{\one}}$ is zero on the even boundary layer $D_{k-1}$, so restricting~\eqref{eq:s} to $D_{k-1}$ forces $\bc_a$ to vanish there.
By Lemma~\ref{lem:2.1}, $a \in \langle {\one} \rangle$. Since $a \neq 0$, we get $a=b={\one}$, a contradiction.
Hence the code defined by \(D_r\cup D_{k-1}\) is minimal.

Adding more layers preserves this conclusion: the evaluation map on
	\(D_r\cup D_{k-1}\) is injective, and any support containment on a
	larger union restricts to the coordinates indexed by this set,
	forcing equality of the parameters. Thus condition~(2) is sufficient.
\smallskip

Finally, we prove necessity. Suppose $S$ contains no even internal layer.

If $S$ consists solely of odd layers, $\bc_{{\one}}$ is the all-one word, whereas $\bc_{e_1}$ is nonzero and is not the all-one codeword. Hence,
\[\mathrm{supp}(\bc_{e_1}) \subsetneq \mathrm{supp}(\bc_{{\one}}).\]

If condition~(2) also fails, it remains to consider a set containing the even boundary layer $k-1$ but no odd internal layer. This can occur only when $k$ is odd, and the only possibilities are $S=\{k-1\}$ and $S=\{1,k-1\}$.

For $S=\{k-1\}$, non-minimality follows from the boundary classification in Theorem~\ref{thm:minimality}.

For $S=\{1,k-1\}$, choose parameters with coordinate-type counts \((p,q,r,h)=(1,2,0,k-3)\), for example
\[\mathrm{supp}(a)=\{4,\dots,k\},\qquad \mathrm{supp}(b)=\{2,\dots,k\}.\]
Both parameters have even weight, and direct evaluation on $D_1$ and $D_{k-1}$ gives
\[\mathrm{supp}(\bc_a) \subsetneq \mathrm{supp}(\bc_b).\]
These cases exhaust all layer sets satisfying neither condition~(1) nor condition~(2), and complete the classification.
\end{proof}
\begin{remark}

The preceding theorem should be compared with three different uses of
simplicial complexes.  The code of Chang and Hyun~\cite{Chang-Hyun} is
a Boolean-function code on the full punctured ambient space; the codes
of Hyun et al.~\cite{Hyun-Lee-lee} use complements
\(\Delta^c\); and Wu and Lee~\cite{Wu-lee} use differences of nested
simplicial complexes.  Consequently, defining sets associated with
simplicial complexes need not themselves be downward closed.

For the present permutation-invariant sets the exact boundary is as
follows.  The set \(D_S\cup\{0\}\) is a simplicial complex if and only
if \(S=\{1,\ldots,r\}\) for some \(r\).  Moreover, \(D_S\) is the
difference of two simplicial complexes if and only if \(S\) is a
consecutive interval.  Indeed, the difference of two nested down-sets
is order-convex in the Boolean lattice, whereas a rank-selected set is
order-convex precisely when its selected ranks form an interval.  A
nonconsecutive \(S\) is therefore neither one simplicial complex nor
one such difference, although it is the disjoint union of the
skeleton differences displayed at the beginning of this section.

Accordingly, the novelty of Theorem~\ref{thm:classfic} is not a new
defining-set mechanism.  It is the uniform if-and-only-if minimality
classification for all permutation-invariant rank-selected sets,
including those with gaps that cannot be represented by one
simplicial complex or by one difference of nested simplicial
complexes.

\end{remark}
\subsection{Dual distance and self-orthogonality}
\begin{definition}
Let \(\varnothing\ne S\subseteq\{1,2,\ldots,k-1\}\), let \(D_S=\bigcup_{t\in S}D_t\), and let \(G_{\mathrm{nat}}\) be the natural generator
	matrix whose columns are indexed by \(D_S\).
 For $t,u,v\in S$, we say that the layer type $(t,u,v)$ is feasible if there exist three distinct columns $x,y,z\in D_S$ of respective weights $t,u,v$ such that
\begin{equation}\label{eq:28}x+y+z=0, \end{equation}
equivalently, there exist subsets $A,B,C\subseteq[k]$ of sizes $t,u,v$, respectively, such that $C=A\triangle B$.\end{definition}
\begin{lemma}
	A layer type $(t,u,v)$ with $1\le t\le u\le v\le k-1$ admits three distinct columns satisfying ~\eqref{eq:28} if and only if
	\begin{equation}\label{eq:29}j=\frac{t+u-v}{2} \end{equation}
	is an integer and
	\begin{equation}\label{eq:30}\max(0,\,t+u-k)\;\le\;j\;\le\;\min(t,u). \end{equation}
\end{lemma}
\begin{proof}
	View $x,y,z$ as indicator vectors of subsets $A,B,C\in D_S$ satisfying \[\#A=t,\quad \#B=u,\quad \#C=v.\] The equation $x+y+z=0$ in $\mathbb{F}_2$ means $C=A\triangle B$. By the inclusion--exclusion principle for symmetric difference,
	\[\#C=\#(A\triangle B)=\#A+\#B-2\#(A\cap B)=t+u-2j,\]
	where $j=\#(A\cap B)$. Hence \[v=t+u-2j,\] which rearranges to ~\eqref{eq:29}.

	 	The elementary bounds for the intersection of a \(t\)-subset and
	 	a \(u\)-subset of \([k]\) give
	 	\[
	 	\max(0,t+u-k)\le j\le\min(t,u).
	 	\]
	
Conversely, if $j$ satisfies~\eqref{eq:29} and~\eqref{eq:30}, choose a $t$-subset $A$, form a $u$-subset $B$ using $j$ elements of $A$ and $u-j$ elements outside $A$, and set $C=A\triangle B$. The three subsets are pairwise distinct, since any equality between two of them would force the third to be empty, whereas $t,u,v\ge1$. Hence, they have sizes $t,u,v$ and satisfy~\eqref{eq:28}.
\end{proof}
\begin{lemma}
	For a feasible layer type $(t,u,v)$, the number of unordered triples $\{x,y,z\}$ of distinct columns with $x+y+z=0$ is
	\begin{equation}\label{eq:31}A_3(t,u,v)=\frac{1}{\mu(t,u,v)}\binom{k}{t}\binom{t}{j}\binom{k-t}{u-j}, \end{equation}
	where $j$ is given by ~\eqref{eq:29} and
	\begin{equation}\label{eq:32}\mu(t,u,v)=
	\begin{cases}
		1, & t,u,v\ \text{pairwise distinct},\\[2pt]
		2, & \text{exactly two of }t,u,v\ \text{are equal},\\[2pt]
		6, & t=u=v.
	\end{cases} \end{equation}
For infeasible types, define $A_3(t,u,v)=0$. Then\begin{equation}\label{eq:33}
		A_3(\C_{k,S}^\perp)
	=\sum_{\substack{t\le u\le v\\ t,u,v\in S}}A_3(t,u,v). \end{equation}
\end{lemma}
\begin{proof}
	Choose $A$ of size $t$ in $\binom{k}{t}$ ways. To form $B$ of size $u$, choose $j$ elements from $A$ in $\binom{t}{j}$ ways and $u-j$ elements from $[k]\setminus A$ in $\binom{k-t}{u-j}$ ways. This constructs an ordered pair $(A,B)$, the set $C=A\triangle B$ is then forced and has size $v$.

		For the prescribed sizes of \(A\) and \(B\), each unordered triple  $\{A,B,C\}$
		is counted once if all three sizes differ, twice if exactly two
		sizes coincide, and six times if all three coincide.
	 Division by $\mu(t,u,v)$ yields ~\eqref{eq:31}.

	Every weight-three dual codeword corresponds uniquely to an unordered triple of distinct columns of $G_{\mathrm{nat}}$ summing to zero. Grouping these triples by their sorted layer type $(t,u,v)$ and summing over $t\le u\le v$ gives~\eqref{eq:33}.
\end{proof}
The following result can be obtained directly from Lemma~\ref{lem:4}.
\begin{lemma}
	Define
	\[{R=\sum_{t\in S}\binom{k-1}{t-1},\qquad P=\sum_{t\in S}\binom{k-2}{t-2}.}\]
	Then every row of $G_{\mathrm{nat}}$ has Hamming weight $R$, and the intersection of the supports of any two distinct rows has size $P$.
\end{lemma}

\begin{theorem}
	Let $\alpha\equiv R\pmod2$ and $\beta\equiv P\pmod2$. Then
	\[\begin{aligned}
		\C_{k,S}\text{ is even} &\iff \alpha=0, \\
		\C_{k,S}\text{ is self-orthogonal} &\iff \alpha=\beta=0, \\
		\C_{k,S}\text{ is doubly even} &\iff
			R\equiv 0\pmod 4\text{ and }\beta=0
	\end{aligned}\]
\end{theorem}
\begin{proof}
	A binary linear code is even if and only if every generating row has even weight, which here means $R\equiv0\pmod2$, or $\alpha=0$. It is self-orthogonal if and only if each row has even weight and every two distinct rows are orthogonal. Since
\[\langle r_i,r_j\rangle=|\operatorname{supp}(r_i)\cap\operatorname{supp}(r_j)|\equiv P\pmod2,\]
this is equivalent to $\alpha=\beta=0$. A doubly even binary code is self-orthogonal, so the necessary conditions are $R\equiv0\pmod4$ and $\beta=0$.
	
		Conversely, if every generating row has weight divisible by four
		and distinct rows are orthogonal, the whole row space is
		self-orthogonal. Applying
		\[
		{\operatorname{wt}}(u+v)
		=
		{\operatorname{wt}}(u)+{\operatorname{wt}}(v)
		-
		2\lvert\operatorname{supp}(u)\cap\operatorname{supp}(v)\rvert
		\]
		inductively shows that every sum of rows has weight divisible by
		four.
\end{proof}
Let $M=G_{\mathrm{nat}}G_{\mathrm{nat}}^T=(M_{ij})_{1\le i,j\le k}$. The preceding row-intersection counts give
\[
M_{ii}=R\bmod2=\alpha,\qquad
M_{ij}=P\bmod2=\beta\quad(i\ne j).
\]
Therefore $M=(\alpha+\beta)I_k+\beta J_k$, where $J_k$ is the all-one matrix.

\begin{lemma}
	Let $\kappa=\dim\ker\phi_{D_S}$ and  $M=(\alpha+\beta)I_k+\beta J_k$.
	 Then
	\begin{equation}\label{eq:38}\dim(\C_{k,S}\cap \C_{k,S}^\perp)=k-\kappa-\rank M. \end{equation}
\end{lemma}
\begin{proof}

		A parameter vector \(a\in\mathbb{F}_2^k\) yields the codeword
		\(\bc_a=\phi_{D_S}(a)\in\C_{k,S}\). Since
		\(\hH(\C_{k,S})
		=\C_{k,S}\cap\C_{k,S}^{\perp}\), we have
		\[
		\bc_a\in\hH(\C_{k,S})
		\iff
		\bc_a\in\C_{k,S}^{\perp}
		\iff
		\phi_{D_S}(a)\text{ is orthogonal to every row of }G_{\mathrm{nat}}.
		\]

	This orthogonality condition is $aM=0$. Moreover, $\ker\phi_{D_S}\subseteq\ker M$, so the hull is naturally isomorphic to $\ker M/\ker\phi_{D_S}$. Since $\dim\ker M=k-\rank M$ and $\dim\ker\phi_{D_S}=\kappa$, Eq.~\eqref{eq:38} follows.
\end{proof}
\begin{theorem}
	The dimension of the hull $\C_{k,S}\cap \C_{k,S}^{\perp}$ is
	\begingroup
	\[\dim(\C_{k,S}\cap \C_{k,S}^\perp)=
	\begin{cases}
		k-\kappa, & (\alpha,\beta)=(0,0),\\
		k-\kappa-1, & (\alpha,\beta)=(1,1),\\
		0, & (\alpha,\beta)=(1,0),\\
		0, & (\alpha,\beta)=(0,1),\ k\text{ even},\\
		1-\kappa, & (\alpha,\beta)=(0,1),\ k\text{ odd},
	\end{cases}\]
	\endgroup
	
\end{theorem}
\begin{proof}
	Take $M=aI_k+bJ_k$ with $a=\alpha+\beta$, $b=\beta$. Then
\begin{itemize}
  \item If $a=b=0$, then $M=0$ and $\rank M=0$.
  \item If $a=0,b=1$, i.e., $(\alpha,\beta)=(1,1)$, then $M=J_k$, and $\rank M= 1$.
  \item If $a=1,b=0$, i.e., $(\alpha,\beta)=(1,0)$, then $M=I_k$, and $\rank M=k$.
  \item If $a=1,b=1$, i.e., $(\alpha,\beta)=(0,1)$, then \(M=I_k+J_k\). If $Mx=0$, then $x=(\sum_i x_i){\one}$. Substitution shows that the kernel is zero when $k$ is even and is $\langle{\one}\rangle$ when $k$ is odd. Hence $\rank M=k$ for even $k$ and $k-1$ for odd $k$.
\end{itemize}	
Substitution into~\eqref{eq:38} completes the proof.
\end{proof}
\subsection{{Illustrative example}}
We now specialize to odd $k\ge5$ and $D_S=D_3\cup D_{k-1}$.
\begin{theorem}
  Suppose $k\ge5$ is odd and $D_S=D_3\cup D_{k-1}$. The associated binary linear code $\C_{k,S}$ is minimal and has parameters $\left[\binom k3+k,\ k,\ d_k\right]$, where
  \[
d_k=
\begin{cases}
6,&k=5,\\
17,&k=7,\\
\dfrac{k(k-1)}2,&k\ge9.
\end{cases}
\]
Moreover, if we take ${\operatorname{wt}}(a)=s$, then
\begin{equation}\label{eq:43}w_s={\operatorname{wt}}(\bc_a)=s\binom{k-s}{2}+\binom s3+
\begin{cases}
s,&s\text{ even},\\
k-s,&s\text{ odd}.
\end{cases}\end{equation}The weight enumerator is
\begin{equation}\label{eq:44}A_{\C_{k,S}}(z)=\sum_{s=0}^k\binom ks z^{w_s}.\end{equation}
\end{theorem}
\begin{proof}
The minimality follows from Theorem~\ref{thm:classfic}.
For \(X\in D_3\), the corresponding evaluation is one precisely
when the integer
\(\lvert\operatorname{supp}(a)\cap {\operatorname{supp}}(X)\rvert\) equals one or three.
Thus, the weight is
\[s\binom{k-s}{2}+\binom{s}{3}.\]
Note that $D_{k-1}=\{\one+e_i:\ 1\le i\le k\}$, the inner product satisfies
\[\langle a,{\one}+e_i\rangle
=(s\bmod2)+a_i.
\] Consequently, the boundary weight is $s$ when $s$ is even, and $k-s$ when $s$ is odd, yielding ~\eqref{eq:43} and ~\eqref{eq:44}.

Take $L_s=s\binom{k-s}{2}+\binom s3$. Since $D_3$ is an odd layer, \[
L_{k-s}=\binom k3-L_s.
\]
Write \(b_s=s\) for even \(s\) and \(=k-s\) for odd \(s\).
We verify that \(L_{k-s}\ge b_s\). For even
\(2\le s\le k-3\),
\[
L_{k-s}\ge(k-s)\binom{s}{2}\ge s.
\]
For \(s=k-1\),
\[
L_1=\binom{k-1}{2}\ge k-1.
\]
For odd \(3\le s\le k-2\),
\[
L_{k-s}\ge(k-s)\binom{s}{2}\ge k-s.
\]
For \(s=1\),
\[
L_{k-1}=\binom{k-1}{3}\ge k-1.
\]
The cases \(s=0,k\) are immediate. Thus
\[
w_s=L_s+b_s\le\binom{k}{3}.
\]

Equality holds only for $s=k$, except that it also holds for $s=1$ when $k=5$. Therefore, the maximum nonzero weight is
\begin{equation}\label{eq:45}
w_{\max} = \binom{k}{3}.
\end{equation}

For $k=5,7$, direct substitution into~\eqref{eq:43} yields the weight sequences
\begin{equation}\label{eq:5-7}
(w_0,\ldots,w_5)=(0,10,8,6,8,10),\  (w_0,\ldots,w_7)=(0,21,22,23,20,17,26,35).
\end{equation}

 For $k \ge 9$, we have \[w_1 = \binom{k-1}{2} + k - 1 = \frac{k(k-1)}{2}.\] If $s > 1$ is odd,  then \begin{equation}\label{eq:49}6(w_s-w_1)
=(s-1)\left[
4s^2+(4-6k)s+3k(k-3)
\right].\end{equation}
We have
\[
\begin{aligned}
f(s)&=4s^2+(4-6k)s+3k(k-3)\\
&\ge f_{\min}:=f\!\left(\frac{3k-2}{4}\right)
=\frac{3k^2-24k-4}{4}>0\qquad(k\ge9).
\end{aligned}
\]
Hence the right-hand side of~\eqref{eq:49} is strictly positive for $k\ge9$, so $w_s>w_1$.

If $s$ is even, set $r=k-s$. Then $r$ is a positive odd integer and
\begin{equation}\label{eq:50}6(w_s - w_1) = s\bigl((s-3)^2+2\bigr) + 3r\bigl(r(s-1)-3s+1\bigr).\end{equation}
Since $(s-3)^2+2>0$, we next divide into three cases to prove $w_s>w_1$.
\begin{itemize}
  \item If $r=1$, then $s=k-1\ge8$, and the right-hand side of~\eqref{eq:50} becomes
  \[
  s(s-1)(s-5)>0.
  \]
  \item If $r=3$, then $s=k-3\ge6$ and   \[
  6(w_s-w_1)
  =s\bigl((s-3)^2+2\bigr)-18
  \ge6(9+2)-18=48>0.
  \]
  \item If $r\ge5$, then
  \[B:=r(s-1)-3s+1=(s-1)(r-3)-2.\]When $s=2$, we have $B=r-5\ge2$ because $k=s+r\ge9$ and $r\ge7$. When $s\ge4$,
  \[
  B\ge3\cdot2-2=4.
  \]Thus $B$ is nonnegative.
\end{itemize}
Therefore $w_s>w_1$ for every even $s>1$. Together with the odd case, this proves that $w_1$ is the minimum nonzero weight for $k\ge9$.
\end{proof}
For $k=5$, the code $\C_{5,\{3,4\}}$ is equivalent to
$\C_{6,4}$ and hence does not give a new inequivalent $[15,5,6]$ code.
Indeed, deleting the last coordinate identifies the even-weight
subspace of $\mathbb F_2^6$ with $\mathbb F_2^5$, the weight-four
vectors whose last coordinate is respectively one or zero map to the
complete weight-three or weight-four layers. Thus the $k=5$ case is
retained only as the first member of the general family.

By~\eqref{eq:5-7}, for $k=5,7$, we have
\[2d_5=12>10,\qquad 2d_7=34<35.\]
For $k\ge9$,
\[
2d_k=k(k-1)\le\frac{k(k-1)(k-2)}6=w_{\max}.
\]
Hence, for every odd \(k\ge 7\), \(\C_{k,S}\) is a minimal
	binary code that violates the Ashikhmin--Barg sufficient condition.
\begin{theorem}\label{thm:mindis}
	Let \(k\ge 5\) be odd and \(D_S=D_3\cup D_{k-1}\).
 Then the dual distance is
\begin{equation}\label{eq:51}
		d(\C_{k,S}^\perp)
	 = \begin{cases}
3, & k = 5, 7, \\
4, & k \ge 9.
\end{cases} \end{equation}
The number of weight-three dual codewords is
\begin{equation}\label{eq:52}
		A_3(\C_{k,S}^\perp)
	= \begin{cases}
15, & k = 5, \\
70, & k = 7, \\
0, & k \ge 9.
\end{cases} \end{equation}
\end{theorem}
\begin{proof}
  Since $S=\{3,k-1\}$, the layer types of three columns can only arise from
  \[
(3,3,3),\ (3,3,k-1),\ (3,k-1,k-1),\ (k-1,k-1,k-1).
\]The types $(3,3,3)$ and $(3,k-1,k-1)$ have total weights $9$ and $2k+1$, respectively. Both are odd and are therefore excluded by the integrality condition~\eqref{eq:29}.

For $(k-1,k-1,k-1)$, although the total weight is even and \[
j=\frac{(k-1)+(k-1)-(k-1)}2=\frac{k-1}{2}
\] is integral, the lower bound ~\eqref{eq:30} requires \[
j\ge (k-1)+(k-1)-k=k-2,
\] which fails for $k \ge 5$ since $(k-1)/2 < k-2$, thus this type is impossible.

Only $(3,3,k-1)$ remains, yielding intersection size \[
j=\frac{7-k}{2}.
\] Since $0 \le j \le 3$, we obtain
\begin{itemize}
\item $k=5, j=1$. We have $\frac{1}{2}\binom{5}{3}\binom{3}{1}\binom{2}{2} = 15$,
\item $k=7, j=0$. We have $\frac{1}{2}\binom{7}{3}\binom{3}{0}\binom{4}{3} = 70$.
\end{itemize}
This proves~\eqref{eq:52} and gives
\(d(\C_{k,S}^{\perp})=3\) for $k=5,7$.

For $k\ge9$, no three columns of the natural generator matrix satisfy~\eqref{eq:29}.
	However, fix distinct coordinates \(i,j\) and choose two distinct
	two-element subsets \(R,R'\) of \([k]\setminus\{i,j\}\). Then
 \[1_{R\cup\{i\}}+1_{R\cup\{j\}}+1_{R'\cup\{i\}}+1_{R'\cup\{j\}}=0.\]
Since the natural generator matrix has nonzero pairwise distinct
columns, the absence of a three-column relation gives
\(d(\C_{k,S}^{\perp})\ge 4\), whereas the four displayed columns give
\(d(\C_{k,S}^{\perp})\le 4\). Therefore,
\[
d(\C_{k,S}^{\perp})=4,
\]
which proves~\eqref{eq:51}.
\end{proof}
\begin{theorem}
		Let \(k\ge 5\) be odd and \(D_S=D_3\cup D_{k-1}\).
Then every row of $G_{\mathrm{nat}}$ has Hamming weight $R$
	\[R=\binom{k-1}{2}+k-1=\frac{k(k-1)}2,\] and the intersection of the supports of any two distinct rows has size  \[P=(k-2)+(k-2)=2(k-2).\]
Consequently,
\begin{align*}\C_{k,S}& \text{ is even} \iff \C_{k,S} \text{ is self-orthogonal} \iff k \equiv 1 \pmod 4,\\
\C_{k,S}& \text{ is doubly even} \iff k \equiv 1 \pmod 8.
\end{align*}
Moreover,
\begin{equation}\label{eq:57} \dim(\C_{k,S} \cap \C_{k,S}^\perp) = \begin{cases}
k, & k \equiv 1 \pmod 4, \\
0, & k \equiv 3 \pmod 4.
\end{cases} \end{equation}
\end{theorem}
\begin{proof}
There are $\binom{k-1}{2}$ triples containing a fixed coordinate, $k-1$ boundary columns containing it, $k-2$ triples containing two fixed coordinates, and $k-2$ boundary columns containing both. Therefore,
\[R=\binom{k-1}{2}+k-1=\frac{k(k-1)}2,\quad P=(k-2)+(k-2)=2(k-2).\]

Since $P$ is always even, evenness and self-orthogonality depend solely on the parity of $R$. For odd $k$, this is equivalent to $k \equiv 1 \pmod 4$. Doubly even requires $R \equiv 0 \pmod 4$, i.e., $k \equiv 1 \pmod 8$.

The Gram matrix is \[G_{\mathrm{nat}}G_{\mathrm{nat}}^{T} \equiv R I_k \pmod 2.\]
	Since \(D_S\) contains \(D_3\), the evaluation map
	\[
	\phi_{D_S}:\mathbb F_2^k\longrightarrow\C_{k,S}
	\]
	is an isomorphism. The hull
	therefore corresponds to \(\ker(G_{\mathrm{nat}}G_{\mathrm{nat}}^T)\).
	Since
	\[
G_{\mathrm{nat}}G_{\mathrm{nat}}^T=(R\bmod 2)I_k,
	\]
	its kernel has dimension \(k\) when \(k\equiv1\pmod4\), and zero
	when \(k\equiv3\pmod4\). This proves~\eqref{eq:57}.
\end{proof}
In particular, the first instance $[15,5,6]$ is self-orthogonal but not doubly even, its hull dimension is $5$, and satisfies
\(d(\C_{k,S}^\perp)=3\) and \(A_3(\C_{k,S}^\perp)=15\).
\section{Graph cut-space codes and defining-set realizations}
\label{sec:graph-cut-space}
\subsection{Classical graph cut-space codes}

Throughout this section, \(G=(V(G),E(G))\) is a finite simple undirected
graph, \(v_G=|V(G)|\), \(n_G=|E(G)|\), and \(c=c(G)\) is the number of
connected components, including isolated vertices.  Let \(B\) be its
binary vertex--edge incidence matrix and
\[
\C_G=\{aB:a\in\mathbb F_2^{v_G}\}.
\]
This is the classical binary cut-space code
\cite[Section~3]{DankelmannKeyRodrigues2013}.  For a connected
graph, the dimension formula is stated in
\cite[Result~1 and Theorem~1]
{DankelmannKeyRodrigues2013}.  Moreover, the proof of that theorem gives
the standard support description
\[
\operatorname{supp}(aB)=\delta_G(\operatorname{supp}(a)),
\qquad
\operatorname{wt}(aB)=\#\delta_G(\operatorname{supp}(a)).
\]
Applying the connected dimension formula componentwise gives
\begin{equation}\label{eq:graph-classical}
\dim\C_G=v_G-c
\end{equation}
for an arbitrary graph with \(c\) connected components.  Since every
codeword has \(2^c\) preimages under \(a\mapsto aB\), it follows that
\[
A_w(\C_G)=2^{-c}\#\{S\subseteq V(G):\#\delta_G(S)=w\}.
\]
By the classical cut--cycle orthogonality theorem
\cite[Theorem~1.9.5, Section~1.9]{Diestel2017}, the dual
\(\C_G^\perp\) is the cycle space, see also
\cite[Section~8]{DankelmannKeyRodrigues2013}.  If \(G\) contains
a cycle, then
\[
d(\C_G^\perp)=\operatorname{girth}(G)
\]
by \cite[Theorem~6]
{DankelmannKeyRodrigues2013}.  Since a three-edge Eulerian subgraph of
a simple graph is necessarily a triangle, this cycle-space description
immediately gives
\[
A_3(\C_G^\perp)=\#\{\text{triangles of }G\}.
\]
If \(G\) is a forest, then its cycle space is zero, and hence
\(\C_G^\perp=\{0\}\).

Recall also that a minimal nonempty cut is a bond, and that in a
connected graph \(\delta_G(S)\) is a bond precisely when both sides of
the corresponding vertex partition induce connected subgraphs
\cite[Section~1.9]{Diestel2017}.  Applying this classical criterion to
every nonzero codeword of \(\C_G\) gives the following result.

\begin{theorem}\label{thm:complete-min}
Assume \(E\ne\varnothing\).  Then \(\C_G\) is minimal if and only if
\[
G\cong K_r\sqcup(v_G-r)K_1
\qquad\text{for some }r\ge2.
\]
\end{theorem}
\begin{proof}
Suppose first that \(G\) has two edge-containing connected components.
Choose vertices \(u\) and \(v\) of positive degree in different such
components.  Since there is no edge between the two components,
\[
\delta_G(\{u\})
\subsetneq
\delta_G(\{u,v\}),
\]
so the corresponding nonzero codewords have strictly nested supports.
Thus \(\C_G\) is not minimal.

Consequently, a minimal \(\C_G\) can have only one edge-containing
component, say \(H\).  If \(H\) is not complete, choose nonadjacent
vertices \(u,v\in V(H)\).  Both vertices have positive degree because
\(H\) is connected and contains an edge.  As \(uv\notin E\),
\[
\delta_G(\{u\})
\subsetneq
\delta_G(\{u,v\}),
\]
which again contradicts minimality.  Hence $H$ must be complete.

Conversely, suppose
\(G\cong K_r\sqcup(v_G-r)K_1\).  Every nonzero cut of \(G\) is the
same edge set as \(\delta_{K_r}(S)\) for some nonempty proper subset
\(S\subset V(K_r)\).  Both \(K_r[S]\) and
\(K_r[V(K_r)\setminus S]\) are connected.  The cited classical bond
criterion therefore shows that every nonzero cut is a bond.  Hence no
nonzero word of \(\C_G\) properly contains the support of another
nonzero word, and \(\C_G\) is minimal.
\end{proof}
This is an application of the cited classical bond criterion, no new
graph-code construction is asserted here.  For the surviving class the
classical complete-graph cut calculation gives
\begin{equation}\label{eq:complete-graph-parameters}
n_G=\binom r2,\qquad
\dim\C_G=r-1,\qquad
\operatorname{wt}(\bc_S)=s(r-s)\quad
(s=\#(S\cap V(K_r))),
\end{equation}
and hence
\[
A_{\C_G}(z)=1+\frac12\sum_{s=1}^{r-1}\binom rs z^{s(r-s)},
\qquad
d(\C_G)=r-1,
\qquad
w_{\max}(\C_G)=\left\lfloor\frac{r^2}{4}\right\rfloor.
\]
The contribution here is the comparison with the binary
Ashikhmin--Barg condition.  We have
\[
\frac{w_{\min}}{w_{\max}}
=\frac{r-1}{\lfloor r^2/4\rfloor}.
\]
This ratio is greater than \(1/2\) for \(2\le r\le6\), equals
\(1/2\) for \(r=7\), and is less than \(1/2\) for \(r\ge8\).
Thus the sufficient condition fails exactly for \(r\ge7\), and the
failure is strict for \(r\ge8\).

The following parity properties are standard consequences of the
incidence-matrix description.  The cut code is even if and only if every
vertex of \(G\) has even degree, hence \(\C_{K_r}\) is even exactly when
\(r\) is odd.

 Moreover, the self-orthogonality criterion in
\cite[Remark~3.1]{MallikYildiz2024} requires
\(\#(E[u]\cap E[v])\) to be even for every pair of distinct
vertices.  If \(uv\) is an edge of a nonempty simple graph, then
\[E[u]\cap E[v]=\{uv\}.\]  Therefore, no nonempty simple-graph cut code is
self-orthogonal. Since every binary doubly-even code is
self-orthogonal, no such cut code is doubly even.  For complete graphs,
the same non-self-orthogonality also follows from
\cite[Lemma~18]{ZhouLiTangDing2019} under the identification
\(\C_{K_r}\simeq\C_{D_2}\).

Finally, the complete-graph hull dimension is already known
\cite[Proposition~2]{GhinelliKeyMcDonough2014}:
\[
\dim\hH(\C_{K_r})=
\begin{cases}
0,&r\text{ odd},\\
r-2,&r\text{ even}.
\end{cases}
\]

\subsection{\texorpdfstring{Graph-function realizations}{Graph-function realizations}}
\label{sec:graph-reinterpretation}


Let $V,W$ be finite-dimensional binary vector spaces, let $D\subseteq V$ and  $F:D\to W$ be a function. Equip $U=V\oplus W$ with a fixed nondegenerate bilinear pairing. Write
\[
U=V\oplus W,\qquad g_x=(x,F(x)),
\]
\[
G=(g_x)_{x\in D},\qquad E=\langle g_x:x\in D\rangle\le U.
\]

Over $\mathbb F_2$, each point represented by a nonzero column $g_x$ has that column as its unique nonzero vector representative.

For a finite point set $P\subseteq\PG(E)$, identify each projective point over $\F_2$ with its unique nonzero vector representative and define the evaluation map
\[
{\phi:U\longrightarrow\F_2^P,\qquad u\longmapsto(\langle u,p\rangle)_{p\in P},}
\]
and write $\C_P:=\operatorname{Im}\phi$. Assuming that $P$ spans $E$,
the following three statements are equivalent, see
\cite{Bonini-Borello,Alfarano-Borello,Tang-Qiu}:
\begin{itemize}
  \item $\C_P$ is minimal;
  \item $\langle P\cap H\rangle=H$ for every hyperplane $H\leq E$;
  \item $P\cap H\subseteq P\cap H'$ implies $H=H'$ for any two hyperplanes $H,H'\le E$.
\end{itemize}
Independently of minimality, the spanning assumption gives the identity
\(\ker\phi=E^\perp\).

\begin{theorem}\label{thm:CBSC-th-2}
Assume $P$ spans $E$. Let $\pi:E\to V$ be linear, put $V_0=\pi(E)$ and $K=\ker\pi$, and choose a linear section $\sigma:V_0\to E$ with $\pi\sigma=\operatorname{id}_{V_0}$. Define
\[
T:E\longrightarrow V_0\oplus K,\qquad
T(z)=\bigl(\pi(z),z-\sigma(\pi(z))\bigr).
\]
Then $T$ is an isomorphism, and the following statements are equivalent:
\begin{enumerate}
\item $\pi|_P$ is injective.
\item Set $D:=\pi(P)\subseteq V_0$, there is a unique function $F:D\to K$ such that
\[
T(P)=\{(x,F(x)):x\in D\}.
\]
\item $\ker\pi\cap\{p+q:p,q\in P,\ p\ne q\}=\varnothing$.
\end{enumerate}
When these conditions hold, $\#P=\#D\le2^{\rank\pi}$ and $P\cap\ker\pi$ contains at most one point. If $E\le V\oplus W'$ and $\pi$ is the restriction to $E$ of the projection onto $V$, then
\[
\pi(E)=\langle D\rangle,\yrev{$\qquad$}
\dim E=\dim\langle D\rangle+\dim\bigl(E\cap(0\oplus W')\bigr).
\]
\end{theorem}
\begin{proof}
The map $T$ is linear and its second component belongs to $K$. Its inverse is
\[
T^{-1}(x,w)=\sigma(x)+w,
\]
so it is an isomorphism. Conditions (1) and (3) are equivalent because, over $\F_2$,
\[
\pi(p)=\pi(q)\quad\Longleftrightarrow\quad p+q\in\ker\pi.
\]
Assume (1). For each $x\in D=\pi(P)$, let $p_x$ be the unique point of $P$ satisfying $\pi(p_x)=x$ and set
\[
F(x):=p_x-\sigma(x)\in K.
\]
Then $T(p_x)=(x,F(x))$, proving (2). Conversely, a graph contains at most one point with any prescribed first coordinate, so (2) implies (1). The same observation proves $\#P=\#D$ and that $P\cap\ker\pi$ has at most one point, the cardinality bound follows from $D\subseteq V_0$.

Finally, if $\pi$ is a coordinate projection and $P$ spans $E$, then
\[
\pi(E)=\langle\pi(P)\rangle=\langle D\rangle.
\]
The dimension follows from rank--nullity, since the kernel of $\pi|_E$ is $E\cap(0\oplus W')$.
\end{proof}

For an $h$-dimensional binary space $E_0$ with basis $\{e_t:t\in I\}$, set
\[
P_0=\{e_t:t\in I\}\cup\{e_t+e_u:t,u\in I,\ t\ne u\}.
\]
This is the $q=2$ tetrahedron construction of \cite[Theorems~5.3 and~5.4]{Alfarano-Borello}, equivalently the defining-set code $\mathcal C_{D_{1,2}}$ and the cut-space code of $K_{h+1}$. Its minimality, parameters and weight distribution are already given there, so we do not restate or reprove them.

We next realize this previously known family in the graph-function framework.


\begin{lemma}
Let $h=2^r\ge8$, $\mathbb{L}=\mathbb{F}_{h}$, $V=\mathbb{F}_{2}\oplus\mathbb{L}\oplus\mathbb{L}$, and let
\begin{align}\label{eq:CBSC-8}
D&=\{(1,t,t^{3}):t\in \mathbb{L}\}
\cup\{(0,t+u,t^{3}+u^{3}):t,u\in \mathbb{L},\ t\neq u\}.
\end{align}Let $E_0=\F_2^{h}$ have standard basis $\{e_{t}\}_{t\in \mathbb L}$. Define
\begin{equation}\label{eq:pi}\pi:E_{0}\longrightarrow V,\qquad \pi(e_{t})=(1,t,t^{3}).\end{equation}
Then
\begin{enumerate}[1]
\item $\pi$ is surjective.
\item There exists a linear section $\sigma:V\to E_{0}$ such that $\pi\sigma=\mathrm{id}_{V}$.
\item $\dim \ker\pi=h-2r-1$.
\end{enumerate}
\end{lemma}
\begin{proof}
  Suppose a linear functional on $V$ vanishes on all $(1,t,t^{3})$. Then there exist $a\in\mathbb{F}_{2}$ and $\beta,\gamma\in\mathbb{L}$ such that
\begin{equation}\label{eq:CBSC-6}
a+\operatorname{Tr}_{\mathbb{L}/\mathbb{F}_{2}}(\beta t+\gamma t^{3})=0\qquad\text{for all }t\in\mathbb{L}. \end{equation}

Taking $t=0$ gives $a=0$. Polarizing ~\eqref{eq:CBSC-6} yields

\[{\operatorname{Tr}_{\mathbb{L}/\mathbb{F}_{2}}\bigl(\gamma(x^{2}y+xy^{2})\bigr)=0\qquad\text{for all }x,y\in\mathbb{L}.}\]

Using the trace identity $\operatorname{Tr}_{\mathbb{L}/\mathbb{F}_{2}}(Ay^{2})=\operatorname{Tr}_{\mathbb{L}/\mathbb{F}_{2}}(A^{2^{r-1}}y)$ and the non-degeneracy of the trace pairing, we obtain
\begin{equation}\label{eq:CBSC-7}\gamma x^{2}+(\gamma x)^{2^{r-1}}=0\qquad\text{for all }x\in\mathbb{L}.\end{equation}

If $\gamma\neq0$, putting $x=1$ in ~\eqref{eq:CBSC-7} gives $\gamma^{2^{r-1}}=\gamma$. Substituting this back into ~\eqref{eq:CBSC-7} yields $\gamma(x^{2}+x^{2^{r-1}})=0$ for every $x\in\mathbb L$. But the polynomial $X^{2}+X^{2^{r-1}}$ is nonzero of degree $2^{r-1}<2^{r}$, so it cannot vanish on all of $\mathbb L$. This contradiction proves that $\gamma=0$, non-degeneracy then gives $\beta=0$. Hence, the image of ~\eqref{eq:pi} has dimension $2r+1$, so $\pi$ is surjective.

 Choose a subset $B\subseteq\mathbb L$ such that $\{(1,t,t^3):t\in B\}$ is a basis of $V$. Define $\sigma$ on this basis by $\sigma(1,t,t^3)=e_t$ for $t\in B$, and extend it linearly. Then $\pi\sigma=\mathrm{id}_V$. Therefore
\[\dim \ker\pi=\dim E_{0}-\dim V=h-(2r+1)=h-2r-1.\qedhere\]
\end{proof}


The restriction $\pi|_{P_0}$ is injective. Indeed, singleton images have first coordinate $1$, whereas pair-sum images have first coordinate $0$. If two pair sums have the same image, then
\[
t+u=v+w=s\ne0,\qquad t^3+u^3=v^3+w^3.
\]
Since $t^3+u^3=s^3+tus$, the sum $s$ and product $tu$ are determined, so the two unordered pairs are the roots of the same polynomial $Z^2+sZ+tu$ and hence coincide. This is the Sidon property of the graph of the Gold APN function $t\mapsto t^3$ in the present coordinates, see \cite{Bracken-Byrne-Markin-McGuire}. Therefore the following definition is well defined.

\textcolor{black}{Define on $D=\pi(P_0)$}
\begin{equation}\label{eq:CBSC-9}F(\pi p)=p-\sigma(\pi(p))\in K\qquad(p\in P_{0}),\end{equation}
where $P_{0}$ contains the standard basis of $E_{0}$, explicitly,
\begingroup
\[\begin{aligned}
F(1,t,t^{3})&=e_{t}+\sigma(1,t,t^{3}),\\
F(0,t+u,t^{3}+u^{3})&=e_{t}+e_{u}+\sigma(0,t+u,t^{3}+u^{3}).
\end{aligned}\]
\endgroup


\begin{proposition}[Graph-function realization]\label{prop:graph-realization}
The linear isomorphism
\[
T:E_0\longrightarrow V\oplus K,\qquad
T(z)=(\pi(z),z-\sigma(\pi(z))),
\]
maps $P_0$ bijectively onto $\{(x,F(x)):x\in D\}$. Consequently, after bases and a coordinate ordering are fixed, the two generator matrices differ by invertible left multiplication and a coordinate permutation. Thus the graph-function construction is equivalent to the known binary tetrahedron code $\mathcal C_{D_{1,2}}$ and does not define a new code-equivalence class.
\end{proposition}
\begin{proof}
The injectivity of $\pi|_{P_0}$ makes $F(\pi(p))=p-\sigma(\pi(p))$ well defined. Theorem~\ref{thm:CBSC-th-2} gives
\[
T(p)=(\pi(p),F(\pi(p)))\qquad(p\in P_0).
\]
Since $P_0$ contains the standard basis of $E_0$, it spans $E_0$. Hence $T(P_0)$ spans $V\oplus K$, and applying $T$ to all columns is precisely an invertible row transformation after bases are chosen. This proves the asserted code equivalence.
\end{proof}

\section{Conclusion}
\label{sec:conclusion}
In this paper, we studied minimality within established binary
	evaluation and incidence-code constructions. For the known complete
	single-layer codes, we recalled the dimension and weight formulas and
	obtained a complete minimality classification. We determined the exact
	extremal weights for \(\C_{k,4}\), showing a strict violation of the
	binary Ashikhmin--Barg sufficient condition for every \(k\ge16\), and
	refined the dual-distance criterion by adding the necessary condition
	\(k\ge3t/2\). The latter also corrects the conclusion for cases such
	as \((k,t)=(8,6)\), whose dual distance is four rather than three.

	For permutation-invariant rank-selected defining sets
	\(D_S=\bigcup_{t\in S}D_t\), we determined the evaluation kernel,
	dimension, and weight enumerator and proved an if-and-only-if
	minimality classification.  Initial layer unions are nonzero parts of
	skeleton complexes, consecutive layer blocks are differences of two
	skeleton complexes, and every arbitrary \(S\) is a disjoint union of
	such differences.  Thus the contribution is the uniform
	classification and parameter analysis over all \(S\), rather than a
	new simplicial-complex or defining-set construction.  The row-Gram
	structure yields uniform criteria for evenness, self-orthogonality,
	the doubly-even property, and hull dimension.  For odd \(k\ge7\), the family
	defined by \(D_3\cup D_{k-1}\) is genuinely nonconsecutive and hence
	cannot be represented by one difference of nested skeleton
	complexes.  Its minimum distance, weight enumerator, dual distance,
	number of weight-three dual codewords, and hull dimension are all
	explicit, and it strictly violates the binary Ashikhmin--Barg
	sufficient condition. 

	For graph cut codes, the classical characterization of bonds
	yields that every nonzero cut is minimal if and only if the graph has
	exactly one edge-containing connected component and that component is
	complete. We expressed this consequence in minimal-code language and
	recorded the associated complete-graph weight data. The cut-space
	dimension, cycle-space dual and complete-graph hull dimensions are
	classical and were included for comparison. The Ashikhmin--Barg
	sufficient condition is violated exactly when \(r\ge7\), with equality
	at \(r=7\) and strict violation for \(r\ge8\).
	
In addition, Section~\ref{sec:graph-cut-space} also places these constructions in the standard
	cutting-blocking-set language and gives an explicit graph-function
	realization of the binary tetrahedron code. The realization is obtained
	by an invertible linear transformation, so it represents the known code
	rather than a new inequivalent code family.
	
	Several directions remain for further study. For a general layer set
	\(S\), determining closed formulas for the minimum and maximum
	nonzero weights of \(\C_{k,S}\) would yield sharper
	Ashikhmin--Barg thresholds. It is also natural to investigate
	\(q\)-ary analogues, for which the binary parity and complement
	arguments used here do not apply directly. Subsets of higher
	Hamming-weight layers may lead to incidence codes associated with
	uniform hypergraphs, where one may seek graph-type classifications of
	minimality, dual distance, and orthogonality. Finally, one may ask
	which cutting blocking sets admit restricted-domain function-graph
	representations that produce genuinely inequivalent code families
	while retaining explicitly computable parameters.

\section*{Acknowledgements}

X. Xie is partially supported by NSFC (Grant No. 62402004) and State Key Laboratory of Cyberspace Security Defense (Grant No. 2025-MS-04).
T. Wu is partially supported by the Natural Science Foundation of Anhui Province (Grant No. 2508085QA033).

\end{document}